\documentclass[3p]{elsarticle}
\makeatletter
\def\ps@pprintTitle{%
 \let\@oddhead\@empty
 \let\@evenhead\@empty
 \def\@oddfoot{\centerline{\thepage}}%
 \let\@evenfoot\@oddfoot}
\makeatother

\date{}
\def\texpsfig#1#2#3{\vbox{\kern #3\hbox{\includegraphics{#1}\kern #2}}\typeout{(#1)}}

\usepackage{url}
\usepackage[hidelinks]{hyperref}
\usepackage{bbm}
\usepackage{eurosym}                           
\usepackage[latin1]{inputenc}                  
\usepackage{url}                               
\usepackage{longtable}                         
\usepackage{array}                             
\usepackage{graphicx,color}
\usepackage{amsthm}
\usepackage{amsbsy}

\usepackage{placeins}  

\usepackage{amssymb}
\usepackage{amsmath}
\usepackage{enumerate}
\usepackage{graphicx,color}
\usepackage{epsfig}
\usepackage[ruled,vlined]{algorithm2e}
\usepackage{multirow,bigdelim}
\usepackage[table]{xcolor}
\usepackage{natbib}
\usepackage{cleveref}

\theoremstyle{plain}
\newtheorem{thm}{Theorem}[section]

\newtheorem{dfn}[thm]{Definition}

\newtheorem{rem}{Remark}[section]

\theoremstyle{remark}

\theoremstyle{plain}
\newtheorem{lem}[thm]{Lemma}

\newtheorem{coroll}[thm]{Corollary}

\theoremstyle{definition}

\newcommand{\e}{{\rm e}}        
\def\R{\mathbb{ R}}             
\def\E{\mathbb{ E}}             
\def\Q{\mathbb{ Q}}             

\def\P{\mathbb{ P}}             

\def\Cov{\mathrm{\mathbb{C}ov}}
\def\corr{\mathrm{corr}}            

\def\Var{\mathrm{\mathbb{V}ar}}   
\renewcommand{\d}{{\rm d}}      

\def\e{{\mathrm{e}}}
\def\1{{\mathbbm{1}}}            

\DeclareMathOperator*{\argmin}{arg\,min}

\theoremstyle{plain}

\usepackage[margin=1cm]{caption}

\numberwithin{equation}{section}	     

\newcommand\red[1]{{\normalcolor#1}}

\title{Cheapest-to-Deliver Collateral: A Common Factor Approach}

\begin{document}

\author[1]{Felix L.\ Wolf}
\ead{felix.wolf@ulb.be}
\author[2,3]{Lech A.\ Grzelak}
\ead{L.A.Grzelak@tudelft.nl}
\author[1]{Griselda Deelstra}
\ead{griselda.deelstra@ulb.be}
\address[1]{Department of Mathematics, Universit\'e libre de Bruxelles, Brussels, Belgium}
\address[2]{Delft Institute of Applied Mathematics, Delft University of Technology, Delft, the Netherlands}
\address[3]{Rabobank, Utrecht, the Netherlands}

\begin{abstract}
\noindent
The collateral choice option gives the collateral posting party the opportunity to switch between different collateral currencies which is well-known to impact the asset price. Quantification of the option's value is of practical importance but remains challenging under the assumption of stochastic rates, as it is determined by an intractable distribution which requires involved approximations. Indeed, many practitioners still rely on deterministic spreads between the rates for valuation.
We develop a scalable and stable stochastic model of the collateral spreads under the assumption of conditional independence. This allows for a common factor approximation which admits analytical results from which further estimators are obtained. We show that in modelling the spreads between collateral rates, a second order model yields accurate results for the value of the collateral choice option. The model remains precise for a wide range of model parameters and is numerically efficient even for a large number of collateral currencies.

\end{abstract}

\maketitle

{\let\thefootnote\relax\footnotetext{The views expressed in this paper are the personal views of the authors and do not necessarily reflect the views or policies of their current or past employers.}}

%
\section{Introduction}  \label{sec:introduction}
%
%
Collateralization describes a market mechanism in which outstanding exposure is covered by low-risk securities, usually in the form of cash or bonds. 
This effectively reduces counterparty credit risk, (\cite{simmons2018collateral}), but adds new challenges and features to the mathematical principles of asset pricing. 
Interest is paid on the posted collateral, determined by the contractually agreed collateral rate. When frictionless collateralization in continuous time is assumed, it can be shown that this collateral rate is exactly the funding rate of the collateralized asset (\cite{PiterbargFunding,Macey2011}). \red{This idea has been fully formalized in \cite{BieleckiRutkowski}, which also feature an analysis of most collateral conventions.}
When the Credit Support Annex (CSA) of the trade allows for different collateral securities, the choice of collateral creates an optionality: the selection of the optimal collateral security, referred to as the \emph{cheapest-to-deliver} (CTD) collateral.
\par
Consequently, this \emph{collateral choice option} should be considered in the pricing of the collateralized asset.
In this article, we introduce a novel approach for the pricing of the collateral choice option which allows for fast and accurate valuation, particularly, when many collateral currencies are available.
\par
Pricing the collateral choice option first found wide attention in the literature in the contributions \cite{FujiiRisk},  \cite{PiterbargCooking,PiterbargSticky}, \cite{PiterbargAntonov}, \cite{SankovichZhu} and \cite{McCloud}.
This happened in the aftermath of the global financial crisis, which makes it plausible to connect the sudden interest in this option with two main changes brought about.
First, the deterioration of interest rates reduced profit margins in the interest rate sector (\cite{Bikker17}), which prompted the search for even small contributions to profitability.
Secondly, cross-currency basis spreads widened during the crisis and afterwards, \red{which pointed to increased differences in the costs of currencies and directly affected the collateral choice option.}
%
%
\par
\smallskip
Throughout the article, we assume that the collateral securities are paid in cash. Then, the associated collateral rates are usually determined by proxies for a minimum-risk interest rate, for example the OIS rate. The model developed here can be applied to other securities, on condition that the associated collateral rate can be expressed in a short rate framework. We further assume that collateralization happens in continuous time without frictions such as payment thresholds.
\red{Discounting practices have been revised in many places after the global financial crisis. Discounting rates need to live up to higher scrutiny as the assumption of risk-free lending and borrowing is reexamined. In the case of collateralized trades, default risks are theoretically removed and discounting only needs to consider the time value of money. Then, the appropriate discounting rate is determined by the posted collateral.}

When there is only one collateral rate, $c_0$, available, \cite{PiterbargCooking} shows that the  measure $\Q_0$\red{, under which the valuation takes place,} is determined by the collateral rate $c_0$. Hence the price of an asset  $V$ with single pay-off $V(T)$ at time $T>0$ is given by 
\begin{equation}
V(0)  = \E^{\Q_0}\Bigl[\exp\bigl({-\int_0^T c_0(s)\d s}\bigr) V(T)\Bigr].
\end{equation}

\par
\smallskip
The collateral choice option arises, when there are multiple currencies $i \in \{0,\dots, N\}$ in which the collateral can be paid. Each of these currencies is associated with it's own collateral rate $c_i$, which are currency specific interest rates. To make these rates directly comparable, they need to be adjusted to act on a common base currency. 
Further, as the  measure is determined by the collateral rate, the presence of multiple collateral rates implies that for each collateral rate $c_i$, there is a  measure $\mathbb{Q}_i$, under which rate $c_i$ is the appropriate discounting rate. 
A consistent model is obtained by fixing one  measure under which all other collateral rates are considered. Typically, the measure chosen is \red{determined by} the domestic currency associated with the collateral rate $c_0$, corresponding to, for example, USD. Then, the dynamics of the remaining rates must be taken under this measure $Q_0$ and, additionally, they must be translated to the domestic currency. 
We directly consider such comparable, FX-adjusted collateral rates under the domestic  measure $\Q_0$ and denote these by $\{r_0, \dots, r_N\}$, where the domestic rate $r_0 = c_0$. 
A more detailed framework permitting different payment and collateral currencies is given in \cite{McCloud}.
\par
The exact value of the collateral choice option depends on the constraints given by the credit support annex (CSA) of the collateralized asset. 
Assuming that the entire collateral account can be switched from one collateral security to another at any time, results in the least restrictive scenario of \emph{full substitution} rights. This scenario can be considered the base case, which allows any collateral posting strategy and as such yields \red{a maximal estimate for the value of the collateral choice option.}
In this broader generalization, the collateral choice option can be detached from the collateralized asset, which is desirable as the valuation methods tends to be too involved to be computed for every collateralized asset individually. This formulation is considered in \cite{PiterbargCooking, PiterbargAntonov,SankovichZhu}.
A more realistic scenario is given under \emph{sticky collateral} rules, where existing collateral cannot be exchanged but newly posted collateral may be paid in the form of a different security.  In this case, the value of the collateral choice option takes dynamics similar to an American option and becomes asset-specific. Mathematically, this scenario extends the question to a (stochastic) control problem and is analysed in \cite{PiterbargSticky}.
\par
In this article, full substitution rights and  an instantaneous exchange of collateral in continuous time are assumed, differentiating the problem under consideration from margin valuation type problems. 
\par
\smallskip
Assume there are $N+1$ currencies available, each associated with an FX-adjusted collateral rate $r_i$, $i\in\{0,\dots,N\}$, in which the collateral can be posted and substituted freely. 
Given this optionality, the optimal strategy, i.e.\ the strategy yielding the highest interest on posted collateral, clearly consists of exchanging all collateral to the highest paying rate at any time. In other words, the valuation of  aforementioned single payment asset $V$ at time $T$ is now
\begin{equation}\label{eq:full-ctd}
\E^{\Q_0}\Bigl[\exp\Bigl({-\int_0^T \max\bigl(r_0(t), \dots, r_{N}(t)\bigr)\d t}\Bigr) V(T)\Bigr],
\end{equation}
with expectation taken under the fixed domestic  measure $\Q_0$.
\par
Accurately, the value of the collateral choice option still depends on each asset, due to interdependencies between the cashflows of the collateralized asset and the collateral rates. To alleviate this situation and obtain a single value, it is common to assume independence between the asset and the collateral rates (\cite{PiterbargCooking}), such that \eqref{eq:full-ctd} can be factored into
\begin{equation}
\E^{\Q_0}\Bigl[\e^{-\int_0^T \max\bigl(r_0(t), \dots, r_{N}(t)\bigr)\d t} V(T)\Bigr]\approx \E^{\Q_0}\Bigl[\e^{-\int_0^T \max\bigl(r_0(t), \dots, r_{N}(t)\bigr)\d t}\Bigr]\,\E^{\Q_0}\Bigl[V(T)\Bigr].
\end{equation}
\red{We emphasize that this is a strong assumption, as the collateralized asset will often not be independent of the collateral rates. This is particularly evident in the case of interest rate products. Nevertheless, it is a pragmatic assumption that practitioners utilize.}
 The resulting \emph{CTD discount factor},
\begin{equation}\label{eq:ctd}
\mathrm{DF}(0,T) = \E^{\Q_0}\Bigl[\exp\Bigl({-\int_0^T \max\bigl(r_0(t), \dots, r_{N}(t)\bigr)\d t\Bigr)}\Bigr],
\end{equation}
offers a first estimate of the collateral choice option's value.
\par
\smallskip
\begin{figure}
\centering
\includegraphics[width=.7\textwidth]{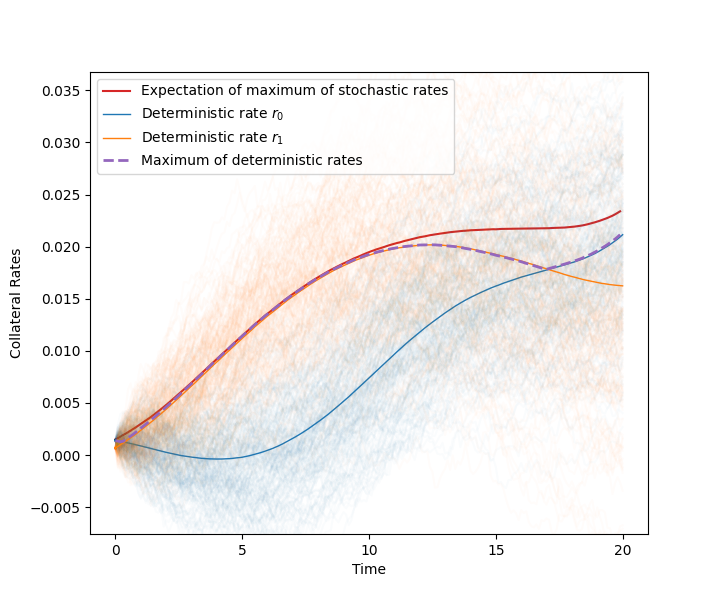}
\caption[]{Trajectories of stochastic and deterministic collateral rates. The expectation of the stochastic rates equals the deterministic rates, but the expectation of the maximum over the stochastic rates is larger than the maximum of the deterministic rates.}
\label{fig:stoch-vs-det}
\end{figure}
\red{
The inherent convexity of the discount factor formula \eqref{eq:ctd} implies that the result obtained from a deterministic model of the (FX-adjusted) collateral rates systematically differs from the result of a stochastic model.
This is illustrated in Figure~\ref{fig:stoch-vs-det}, where an example of deterministic and stochastic collateral rates is given, chosen such that the expectation of the stochastic rates equals the deterministic rates. As a direct consequence of Jensen's inequality, the maximum of the stochastic rates is always greater or equal to the maximum of the deterministic rates,
\begin{equation}
\E^{\Q_0}\Bigl[ \max\bigl(r_0(t), \dots, r_N(t)\bigr)\Bigr] \geq  \max\Bigl(\E^{\Q_0}\bigl[r_0(t)\bigr], \dots, \E^{\Q_0}\bigl[r_N(t)\bigr]\Bigr).
\end{equation}
When these maxima are non-negative, as in the collateral spread formulation introduced in the following section, the discount factor produced by the deterministic model is always greater or equal to the discount factor obtained from the stochastic model.}

\subsection{From collateral rates to collateral spreads}
A difficulty in the stochastic modelling of \eqref{eq:ctd} lies in the low analytical tractability of the distribution of the integral of the maximum of collateral rates, $\int_0^T \max(r_0(t), \dots, r_N(t))\d t$. 
In \cite{SankovichZhu}, an approximation scheme is developed, which matches an analytically tractable distribution, a quadratic polynomial of the normal distribution, to the first three moments of the integral of the maximum. This method focuses in particular on the preservation of the third moment of the integral, as the authors notice a significant skewness in the distribution of the integral of the maximum of FX-adjusted collateral rates. To obtain the three required moments of the integral of the maximum, an extension of the Clark procedure (\cite{Clark}) for the maximum of normal distributions is developed.
\par
In this article, we consider a model based on \emph{collateral spreads}, which are the difference processes between FX-adjusted collateral rates. We will show in this article, that they can be captured accurately by a model based on only two moments.
The collateral spread approach was introduced in \cite{PiterbargCooking}, where the differences between the collateral rates and a selected base collateral rate are modelled directly. 
To this end, we fix the base collateral rate $r_0$ and define the \emph{collateral spreads} $q_i$, $i\in\{1,\dots,N\}$ by 
\begin{equation}\label{eq:spread-def}
q_i := r_i - r_0.
\end{equation}
This allows for the transformation of the maximum of the collateral rates into a maximum of spreads, 
\begin{equation}
\max\bigl(r_0(t), \dots, r_{N}(t)\bigr) = r_0(t) + \max\bigl(0, q_1(t), \dots, q_N(t)\bigr).
\end{equation}
Applying this transformation to \eqref{eq:ctd} yields
\begin{align}
\mathrm{DF}(0,T) &= \E^{\Q_0}\Bigl[\exp\Bigl({-\int_0^T r_0(t) + \max\bigl(0, q_1(t), \dots, q_{N}(t)\bigr)\d t}\Bigr)\Bigr] \nonumber \\
&= P_0(0, T) \E^{T}\Bigl[\exp\Bigl({-\int_0^T  \max\bigl(0, q_1(t), \dots, q_{N}(t)\bigr)\d t}\Bigr)\Bigr], \label{eq:ctd-spread}
\end{align}
where $\E^T$ denotes expectation under the $T$-forward measure with num\'eraire $P_0(0,T):=\E^{\Q_0}[\exp(-\int_0^T r_0(t)\d t)]$.
In this setting, typically the spreads $q_i$, not the collateral rates, are provided with stochastic dynamics. Whilst the collateral spread maximum in \eqref{eq:ctd-spread} differs from the collateral rate maximum in \eqref{eq:ctd} by an additional constant zero component, the problem of low tractability persists.
\par
\red{In \cite{PiterbargAntonov}, the special case of exactly two collateral currencies is covered. Besides approximation schemes based on the first and second-order Taylor expansions, a conditional independence approach for the valuation of the collateral choice option is introduced, in which the distributions of the maximum process at different times 
are assumed conditionally independent.
}
\par
In this article, we develop an alternative approximation scheme for collateral spreads suitable to the case when more than two collateral currencies are available.
\par
\red{From a practical and risk management perspective, this is an essential advance over approximation of the multi-currency option with a carefully calibrated two-currency setup, as contributions of the individual currencies can be separated. 
Furthermore, it is not given that the high-dimensional multi-currency problem can generally be projected down to the one-dimensional problem solved by a two-currency spread setup.}
\par

\red{The collateral spread method proposed in this article}
is based on the observation that the volatilities of the collateral spreads $q_i$, $i\in\{1,\dots,N\}$ are significantly smaller than the volatilities of the FX-adjusted collateral rates considered in \cite{SankovichZhu}. This can be attributed to high correlations between the FX-adjusted collateral rates, which results in low volatility of their difference processes, the collateral spreads.
As a result, the third moment of the integral $\int_0^T \max(0, q_1(s), \dots, q_N(s))\d s$ is significantly smaller than the third moment of the integral $\int_0^T \max(r_0(s), \dots, r_N(s))\d s$. This allows us to create a stochastic model based on only two moments of the underlying distribution, which achieves high precision within a short computation time, even in the presence of a large number of available collateral currencies.
\par
\red{In this procedure, the marginal distributions of the maximum $\max(0, q_1(t), \dots, q_N(t))$ at each time $t\in(0,T]$ are approximated by a conditionally independent model. This allows for direct computation of the first two moments, $\E[\max(0, q_1(t), \dots, q_N(t))]$ and $\Var[\max(0, q_1(t), \dots, q_N(t))]$, on which the second-order model is built. It is important to note differences to the conditional independence assumption proposed in \cite{PiterbargAntonov}. Here, we assume that at fixed times the components of the maximum, the collateral spreads, are conditionally independent from another. \cite{PiterbargAntonov}, however, assume conditional independence between the maximum process distributions at different times to approximate the process structure. In the approach presented here, this particular obstacle is treated by two different estimators, which approximate the behaviour of the maximum process along time with a standard It\^o diffusion and a mean-reverting process, respectively.}
\par
In Section~\ref{sec:model}, we introduce the underlying collateral spread model and the first and second-order approximations of the CTD discount factor used. In Section~\ref{sec:First-order}, a conditional independence model called the \emph{common factor approximation} is defined, which directly results in a first order estimator of the CTD discount factor. Section~\ref{sec:Second-order} extends this result by two types of second-order estimators based on the common factor approximation. Section~\ref{sec:extension} and Appendix~\ref{sec:appx1} outline a model extension which supports more general correlation structures between the collateral currencies.
In Section~\ref{sec:example} we consider accuracy, model parameter sensitivity\red{, robustness in a stressed scenario} and speed of these first and second-order approximations of the collateral choice option.
%
%
%
%
\section{The CTD model}\label{sec:model}
\subsection{Collateral Spread Model}\label{sec:coll-spread-model}
We model the collateral spreads $q_i$, $i\in\{1,\dots,N\}$ between the FX-adjusted collateral rates $r_i$ and the designated base collateral rate $r_0$. These spreads are directly modelled under the $T$-forward measure, where $P_0(0, t) = \E^{\Q_0}[\exp(-\int_0^t r_0(s)\d s)]$ is the num\'eraire. In the following, all probabilities, expectations and variances are taken with respect to this measure.

The spreads $q_i$ are modelled with Hull--White dynamics of the form
\begin{equation}\label{eq:dHW}
\d q_i(t) = \kappa_i \bigl(\theta_i(t) - q_i(t)\bigr) \d t + \xi_i \d W_i(t), i \in \{1,\dots, N\},\ t\in(0,T],
\end{equation}
with initial values $q_i(0)\in\R$, speed of mean reversion parameters $\kappa_i\in\R^+$, long-term means $\{\theta_i(t)\in\R\colon t \in (0,T]\}$ and volatility parameters $\xi_i\in\R^+$. The driving stochastic processes $W_i$ are Brownian motions under the $T$-forward measure and correlated such that their quadratic covariation is $\d [W_i, W_j]_t = \rho_{i,j}\d t$ for some value $\rho_{i,j}\in[0,1)$ for $i,j\in\{1,\dots,N\}$ and $i \neq j$. 
From the well-known solution of \eqref{eq:dHW}, given in \cite{GrzelakOosterlee},
\begin{equation}\label{eq:HWsolution}
q_i(t) = q_i(0)\e^{-\kappa_i t} + \kappa_i\int\limits_0^t \theta_i(s)\e^{-(t-s)\kappa_i}\d s + \xi_i \e^{-\kappa_i t} \int\limits_0^t \e^{\kappa_i s}\d W_i(s),
\end{equation}
it is clear that, at each time $t\in(0,T]$, the spread $q_i(t)$ is normally distributed with expectation
\begin{equation}\label{eq:HW-expect}
\mu_i(t) := \E[q_i(t)] = q_i(0)\e^{-\kappa_i t} + \kappa_i\int\limits_0^t \theta_i(s)\e^{-(t-s)\kappa_i}\d s,
\end{equation}
and variance
\begin{equation}\label{eq:HW-var}
\sigma_i^2(t) := \Var[q_i(t)] = \frac{\xi_i^2}{{2\kappa_i}}{\bigl(1-\e^{-2\kappa_i t}\bigr)}.
\end{equation}
Further, at any $t\in(0,T]$ the  linear correlation between two different spreads $q_i(t)$ and $q_j(t)$ is given by
\begin{align}\label{eq:HW-corr}
\corr(q_i(t), q_j(t)) &:= \frac{\Cov\bigl[q_i(t), q_j(t)\bigr]}{\bigl(\Var[q_i(t)] \Var[q_j(t)]\bigr)^{1/2}} \nonumber \\
&=  2 \rho_{i,j} \frac{\sqrt{\kappa_i\kappa_j}}{\kappa_i + \kappa_j} \frac{1-\e^{-(\kappa_i+\kappa_j)t}}{\sqrt{(1-\e^{-2\kappa_i t})(1-\e^{-2\kappa_j t})}}.
\end{align}
\begin{rem}
The correlation of Hull--White processes $q_i$ and $q_j$ is only constant in time when they have the same speed of mean reversion  parameter, $\kappa_i = \kappa_j$. Otherwise, the correlation slowly decreases in time. 
\end{rem}

\subsection{Approximating the CTD}\label{sec:approxCTD}
The cheapest-to-deliver problem in \eqref{eq:ctd-spread} is characterized by the term 
\begin{equation}\label{eq:ctd-writtenout}
\E\Bigl[\exp\bigl({-\int_0^T \max(0, q_1(t), \dots, q_N(t))\d t}\bigr)\Bigr],
\end{equation}
for which no analytical solution is known when the spreads $q_i$ are modelled with Hull--White dynamics. 
In the following, we will introduce approximations which lead to analytically tractable formulas.
First, we establish some additional notation.
\begin{dfn}
Let $M$ be the maximum process defined for every $t\in(0,T]$ as
\begin{equation}
M(t) = \max\bigl(0, q_1(t), \dots, q_N(t)\bigr).
\end{equation}
Let $Y(T)$ be the integral of the maximum process until maturity $T$,
\begin{equation}
Y(T) = \int_0^T M(t) \d t.
\end{equation}
\end{dfn}
The object of interest is the first moment of the random variable $\exp(-Y(T))$, which is equivalent to the expression in \eqref{eq:ctd-writtenout}. 
An approximation of this first moment can be found by the delta method, in which the expectation $\E[\exp(-Y(T))]$ is replaced by the expectation of a polynomial approximation of $\exp(-Y(T))$. Applications of this method can be found in \cite{Grzelak2011}, \cite{amstrup2005handbook} and an analytical study of the method is given in \cite{oehlert1992}.
The Taylor approximation of $\exp(-Y(T))$ of order $\ell$ around expansion point $\E[Y(T)]$ is given by 
\begin{equation}\label{eq:fulltaylor}
\e^{\E[-Y(T)]}+ \sum\limits_{m=1}^\ell \frac{(-1)^m}{m!} \e^{\E[-Y(T)]} \bigl(Y(T)-\E[Y(T)]\bigr)^m,
\end{equation}
and depends on central moments of the integral, $\E[(Y(T) - \E[Y(T)])^m]$, of orders $m\in\{1,\dots,\ell\}$.
We focus on approximation schemes up to the second order, as these already yield very precise results which is shown in Section~\ref{sec:example}.
\par
\medskip
In the first order Taylor approximation, it holds 
\begin{align}
\E\bigl[\exp\bigl(-Y(T) \bigr)\bigr] &\approx \E\Biggl[\exp\Bigl(\E\bigl[-Y(T) \bigr]\Bigr) - \exp\Bigl(\E\bigl[-Y(T) \bigr]\Bigr) \bigl(Y-\E[Y(T)]\bigr)\Biggr] \nonumber \\
&=\exp\Bigl(-\E\bigl[Y(T) \bigr]\Bigr). \label{eq:exchange-exp}
\end{align}
That is, the first order approximation is equivalent to exchanging the exponential and expectation operators.
Analogously, the second-order approximation is given by 
\begin{align}
\begin{split}
\E\bigl[\exp\bigl(-Y(T) \bigr)\bigr] &\approx \E\Bigl[\exp\bigl(\E[-Y(T)]\bigr) - \exp\bigl(\E[-Y(T)]\bigr) \bigl(Y-\E[Y(T)]\bigr) \nonumber \\
  &\hphantom{\E\Bigl[} + \frac12 \exp\bigl(\E[-Y(T)]\bigr) \bigl(Y-\E[Y(T)]\bigr)^2 \Bigr]
\end{split}\\
&=\exp\bigl(-\E[Y(T)]\bigr)\Bigl( 1 + \frac12 \Var[Y(T)] \Bigr). \label{eq:2ndtaylor}
\end{align}
\par
We use the first and second-order approximations in \eqref{eq:exchange-exp} and \eqref{eq:2ndtaylor}, respectively, to estimate the cheapest-to-deliver discount factor in \eqref{eq:ctd-spread}.
\par
The first order model is introduced in Section~\ref{sec:First-order}. Its essential term $\E[Y(T)]$ is obtained from a common factor approximation, which admits a semi-analytical expression by the use of conditional independence.
\par
Two different second-order models are considered in Section~\ref{sec:Second-order}. Therein, the additionally required term $\Var[Y(T)]$ is approximated by two different estimators, both based on the previously established common factor model. The first estimator, denoted diffusion estimator, is treated in Section~\ref{subsec:diffusion} and the second estimator, denoted mean-reverting estimator, in Section~\ref{subsec:meanreverting}.
\section{First order model with a common factor approximation}\label{sec:First-order}

We need to find the first moment of the integral $Y(T)$. By use of Fubini's theorem, it is possible to exchange integration and expectation to obtain
\begin{equation}
\E\bigl[Y(T)\bigr] = \E\Bigl[\int_0^T M(t) \d t\Bigr]  = \int_0^T \E\bigl[M(t)\bigr] \d t.
\end{equation}
In consequence, the expectation $\E[M(t)]$ can be considered individually at every time $t$. As the model is intended for numerical implementation, we define it directly over a time discretization  $\mathcal{T} = \{t_k\colon 0 \leq k \leq R\}$ of $[0,T$] with  $R\in\mathbb{N}$ many steps and uniform step size $\Delta t = T/R$.
\par
\medskip
The first order approximation is thus obtained from a pointwise evaluation of $\E[M(t_k)]$ at each time $t_k \in \mathcal{T}$. For each such time, a \emph{common factor approximation} of the random vector $(q_1(t_k), \dots, q_N(t_k))$ is defined, which admits a semi-analytical solution.

The common factor approximation belongs to the class of Gaussian copula models, which are widely used in financial applications 
(\cite{Meissner2014, cherubini2004copula}). 
In a Gaussian copula model, components $x_i$ are weighted sums of independent normal random variables. Exemplarily, $x_i$ could be expressed as the sum of a shared random variable $C$ which appears in every component with a component-specific weight $\alpha_i$, and an individual random variable $A_i$, which is specific to the component:
\begin{equation}
x_i = \alpha_i C + A_i.
\end{equation}
This shared random variable $C$ is often called the \emph{common factor}. It is central to our model that each component places equal weight $\alpha_i = 1$ on the common factor, as this enables the core argument in \eqref{eq:indep-max}. To emphasize this fact, we underline the importance of the common factor in the name of the approximation.

\subsection{Common factor approximation}\label{subsec:cf}
At every time $t_k\in\mathcal{T}$ the spreads $(q_1(t_k), \dots, q_N(t_k))$ form a multivariate normal random vector. We introduce a \emph{common factor approximation} $(\widetilde q_1(t_k), \dots, \widetilde q_N(t_k))$ of this random vector which allows us to compute moments of the maximum of these approximations and zero, $\widetilde M(t_k) = \max(0, \widetilde q_1(t_k), \dots, \widetilde q_N(t_k))$. These moments approximate the desired moments of $M(t_k) = \max(0, q_1(t_k), \dots, q_N(t_k))$. 
\par
The use of a common factor approximation is motivated by the fact that the distribution of the maximum of correlated normal distributions is computationally demanding. A sufficiently general analytical solution was only recently obtained in \cite{Nadarajah2018} and still depends on infinite sums. On the other hand, the distribution of the maximum of independent random variables is simply the product of the involved marginal distributions as recalled in Lemma~\ref{lem:max-indep}. 
However, market observations show that the involved collateral spreads are correlated. Therefore, we assume conditional independence, which lets us express the required maximum in terms of independent components, leading to a model which allows for correlations between the spreads while admitting a fast solution.
\par
The common factor approximation is defined pointwise at each time $t_k \in \mathcal{T}$, and must be repeated at any such time in the time discretization.
\par
We formalize the model in the following definition.
\begin{dfn}[Common factor approximation]\label{def:cf}
Let $q_i(t_k) \sim \mathcal{N}(\mu_i(t_k), \sigma_i^2(t_k))$, $i\in\{1,\dots,N\}$ be normal random variables. 
The \emph{common factor approximation} of $(q_1(t_k), \dots, q_N(t_k))$ is given by the random vector $(\widetilde q_1(t_k), \dots, \widetilde q_N(t_k))$, defined componentwise as
\begin{equation}
\widetilde q_i(t_k) = C(t_k) + A_i(t_k),\ i\in\{1,\dots, N\},
\end{equation}
where $C(t_k), A_1(t_k), \dots, A_N(t_k)$ are independent random variables with distributions
\begin{align}
C(t_k) &\sim \mathcal{N}\bigl(0, \sigma_{\min}^2(t_k) |\gamma(t_k)|  \bigr),\ \gamma(t_k) \in [0,1),   \label{eq:C}  \\
A_i(t_k) &\sim \mathcal{N}\bigl(\mu_i(t_k), \sigma_i^2(t_k) - \Var[C(t_k)]\bigr),\ i\in\{1,\dots,N\}.      \label{eq:Ai}
\end{align}
Here, $\sigma_{\min}^2(t_k) = \min(\sigma_1^2(t_k),\dots, \sigma_N^2(t_k))$ denotes the smallest variance occurring among the random variables $q_i(t_k)$.
\end{dfn}
The parameter $\gamma(t_k)\in[0, 1)$ is called the \emph{correlation optimization parameter} of the common factor approximation and controls the correlation structure of the common factor approximation.
\begin{lem}[Properties of the common factor approximation]\label{lem:cfprops}
Let $(\widetilde q_1(t_k),\dots, \widetilde q_N(t_k))$ be the common factor approximation given in Definition~\ref{def:cf}. 
\begin{enumerate}
\item The marginal distributions of the common factor approximation $(\widetilde q_1(t_k),\dots, \widetilde q_N(t_k))$ and the random variables $(q_1(t_k),\dots, q_N(t_k))$ coincide:
\begin{equation}
\widetilde q_i(t_k) \overset{d}= q_i(t_k) \sim \mathcal{N}(\mu_i(t_k), \sigma_i^2(t_k)),\ i\in\{1,\dots,N\}.
\end{equation}
\item The correlations of the common factor approximation depend only on the correlation optimization parameter $\gamma(t_k)$ and the standard deviations $(\sigma_1(t_k), \dots, \sigma_N(t_k))$:
\begin{equation}\label{eq:cf-corr}
\corr\bigl(\widetilde q_i(t_k), \widetilde q_j(t_k)\bigr) = \frac{\sigma_{\min}^2(t_k) |\gamma(t_k)|} { \sigma_i(t_k) \sigma_j(t_k)},\ i\neq j.
\end{equation}
\item The correlations of the common factor approximation are bounded from above by the ratio
\begin{equation}
0 \leq \corr(\widetilde q_i(t_k), \widetilde q_j(t_k)) < \frac{\sigma_{\min}^2(t_k)} { \sigma_i(t_k) \sigma_j(t_k)}.
\end{equation}
\end{enumerate}
\end{lem}
In particular, this lemma shows that the correlations of the common factor approximation are different from the correlations of the approximated random variables, $\corr(q_i(t_k), q_j(t_k))$. They only depend on the correlation optimization parameter $\gamma(t_k)$ which has to be chosen appropriately.
\begin{proof}
\begin{enumerate}
\item By definition, the marginal distributions of the components $\widetilde q_i(t_k)$, $i\in\{1,\dots,N\}$ are sums of normal distributions and thus normally distributed. It further holds for all $i\in\{1,\dots,N\}$:
\begin{equation}
\E\bigl[\widetilde q_i(t_k)\bigr] = \E\bigl[C(t_k) + A_i(t_k)\bigr] = \E\bigl[A_i(t_k)\bigr] = \mu_i(t_k),
\end{equation}
and by independence of $C(t_k)$ and $A_i(t_k)$, 
\begin{equation}
\Var\bigl[\widetilde q_i(t_k)\bigr] = \Var\bigl[C(t_k)\bigr] + \Var\bigl[A_i(t_k)\bigr] = \sigma_i^2(t_k).
\end{equation}
\item For all $i\neq j \in\{1,\dots,N\}$ it holds by independence of $C(t_k)$, $A_i(t_k)$ and $A_j(t_k)$ that
\begin{align}
\corr(\widetilde q_i(t_k), \widetilde q_j(t_k)) &= \frac{\Cov[C(t_k) + A_i(t_k), C(t_k) + A_j(t_k)]}{\bigl(\Var[\widetilde q_i(t_k)]\Var[\widetilde q_j(t_k)]\bigr)^{1/2}} \nonumber \\
&= \frac{\Var[C(t_k)]}{\bigl(\Var[\widetilde q_i(t_k)]\Var[\widetilde q_j(t_k)]\bigr)^{1/2}} = \frac{\sigma_{\min}^2(t_k)|\gamma(t_k)|}{\sigma_i(t_k) \sigma_j(t_k)}. 
\end{align}
\item The correlation bounds follow immediately from the bounds on the correlation parameter $\gamma(t_k)$. To see that $|\gamma(t_k)|<1$ is a necessary condition, assume that $q_j(t_k)$, $1\leq j\leq N$ is the random variable with minimal variance, $\sigma_j^2(t_k) = \sigma_{\min}^2(t_k)$. For the variance $\Var[A_j(t_k)]$ to be positive, it must hold that $|\gamma(t_k)|<1$.
\end{enumerate}
\vspace{-1em}
\end{proof}
The correlation matrix of the spreads at time $t_k$, denoted by $R(t_k)$, with entries
\begin{equation}
R_{i,j}(t_k) = \corr\bigl(q_i(t_k), q_j(t_k)\bigr),\ {i,j\in\{1,\dots, N\}},
\end{equation}
is determined by the correlation formula for Hull--White processes in \eqref{eq:HW-corr}. 
In general, the correlation matrix of the common factor approximation, $\widetilde R(t_k)$, with entries 
\begin{equation}
\widetilde R_{i,j}(t_k) = \corr\bigl(\widetilde q_i(t_k), \widetilde q_j(t_k)\bigr),\ {i,j\in\{1,\dots, N\}},
\end{equation}
which are determined by \eqref{eq:cf-corr}, need not coincide with the correlation matrix $R(t_k)$
for any value of the correlation optimization parameter $\gamma(t_k)$. 
\par
However, an optimal parameter $\gamma^*(t_k)$ can be found which minimizes the distance between these two correlation matrices. If this distance is chosen as the Frobenius matrix norm, the minimization becomes equivalent to the convex optimization problem
\begin{align}\label{eq:frobnorm}
\min\limits_{\gamma(t_k)} \Biggl(\sum\limits_{i=1}^N \sum\limits_{j=1}^N \Bigl(\corr\bigl(\widetilde q_i(t_k),\widetilde q_j(t_k)\bigr) - \corr\bigl(q_i(t_k),  q_j(t_k)\bigr)\Bigr)^2 \Biggr)^{1/2}.
\end{align}

In order to find the optimal parameter $\gamma^*(t_k)$, we apply a numerical solver to
\begin{align}\label{eq:optimalgamma}
\gamma^*(t_k) =\argmin_{\gamma(t_k)} \Biggl(\sum\limits_{i=1}^N \sum\limits_{j=1}^N \Bigl(\frac{\sigma_{\min}^2(t_k) |\gamma(t_k)|} {\sigma_i(t_k) \sigma_j(t_k)} - \corr\bigl(q_i(t_k),  q_j(t_k)\bigr)\Bigr)^2 \Biggr)^{1/2}.
\end{align}

\par
In the important case of three available collateral currencies, one currency is encoded in the zero component of the maximum and $N=2$ spreads remain, with one correlation  to consider between them. Then, at every time $t_k$, the common factor correlation matrix $\widetilde R(t_k)$ can be matched exactly to the spread correlation matrix $R(t_k)$, by setting the correlation optimization parameter $\gamma(t_k)$ to
\begin{equation}
\gamma(t_k) = \corr\bigl(q_1(t_k), q_2(t_k)\bigr) \frac{\max\bigl(\sigma_1(t_k), \sigma_2(t_k)\bigr)}{\min\bigl(\sigma_1(t_k), \sigma_2(t_k)\bigr)}.
\end{equation}
Within the bounds given in Lemma~\ref{lem:cfprops} part $3$, this choice ensures equal correlations
\begin{equation}
\corr\bigl(q_1(t_k), q_2(t_k)\bigr) = \corr\bigl(\widetilde q_1(t_k), \widetilde q_2(t_k)\bigr).
\end{equation}
\par
\medskip
We recall that the distribution of the maximum $M(t_k) = \max(0, q_1(t_k), \dots, q_N(t_k))$ at time $t_k$ is not suitably tractable. By replacing the random variables $q_i(t_k)$ with the components of the common factor approximation, a semi-analytical expression is obtained. To this end, define the \emph{common factor maximum}
\begin{equation}\label{eq:cfmax}
\widetilde{M}(t_k) = \max\bigl(0, \widetilde q_1(t_k), \dots, \widetilde q_N(t_k)\bigr).
\end{equation}
The common factor maximum can be rewritten as
\begin{align}
\widetilde M(t_k) &= \max\bigl(0, C(t_k) + A_1(t_k), \dots, C(t_k) + A_N(t_k)\bigr) \nonumber \\
&= C(t_k) + \max\bigl(-C(t_k), A_1(t_k), \dots, A_N(t_k)\bigr), \label{eq:indep-max}
\end{align}
where the last term is a maximum over independent random variables. In the following lemma, we recall that the cumulative distribution function of such a maximum is available in closed form. 
With the cumulative distribution function at hand, arbitrary moments, including the expectation, can be computed.
In particular, an expression for the common factor maximum $\E[\widetilde M(t_k)]$ is found, which approximates the expectation of the maximum of the spreads, $\E[M(t_k)]$.
\begin{lem}[The maximum of independent random variables]\label{lem:max-indep} 
Let $Z_1, \dots, Z_N$ be independent random variables. The cumulative distribution function of the maximum over all the random variables $Z_i$, $i\in\{1,\dots,N\}$ is given by
\begin{equation}\label{eq:lem-max-cdf}
\P[\max(Z_1, \dots, Z_N) \leq x]  = \prod\limits_{i=1}^N \P[Z_i \leq (x)].
\end{equation}
\end{lem}
\begin{proof}
It holds $\P[\max(Z_1, \dots, Z_N) \leq x] = \P[Z_1 \leq x, \dots, Z_N \leq x] = \prod\limits_{i=1}^n \P[Z_i \leq x]$ which factorizes by independence.
\end{proof}
\par
Applying Lemma~\ref{lem:max-indep} to $\max(-C(t_k), A_1(t_k), \dots, A_N(t_k))$ in \eqref{eq:indep-max} yields that its distribution function, denoted by  $F_k$, is given by
\begin{equation}\label{eq:max-cdf}
F_{k}(x) := \Phi\Biggl(\frac{x}{\sqrt{\sigma_{\min}^2(t_k)|\gamma(t_k)|}}\Biggr) \prod\limits_{i=1}^N \Phi\Biggl(\frac{x-\mu_i(t_k)}{\sqrt{\sigma_i^2(t_k) - \sigma_{\min}^2(t_k) |\gamma(t_k)|}}\Biggr),
\end{equation}
where $\Phi$ denotes the standard normal cumulative distribution function.
The following lemma recalls the well known result that the cumulative distribution function can be used to compute moments of a random variable.
\begin{lem}[Moments from the CDF]\label{lem:cdfmoments}
Let $Z$ be a random variable with cumulative distribution function $F_Z\colon \R \to [0,1]$ and let $m\in\mathbb{N}$. When the $m^\text{th}$ moment of $Z$ is finite, it is given by
\begin{equation}
\E[Z^m] = \int_{-\infty}^0 (-1)^m m x^{m-1} F_Z(x) \d x + \int_0^\infty m x^{m-1} \bigl(1-F_Z(x)\bigr)\d x.
\end{equation}
\end{lem}
\medskip
By setting $m=1$, Lemma~\ref{lem:max-indep} and Lemma~\ref{lem:cdfmoments} yield an expression for the expectation of the common factor maximum. 
\begin{coroll}[Expectation of the common factor maximum]\label{cor:expMtilde}
The expectation of the common factor maximum $\widetilde M(t_k)$ at times $t_k \in \mathcal{T}$ is given by
\begin{align}
\E[\widetilde M(t_k)] &= \E\bigl[\max(0,\widetilde q_1(t_k), \dots, q_N(t_k))\bigr] \nonumber \\
&=\int_{-\infty}^0 -F_k(x)\d x + \int_0^\infty (1-F_k(x)) \d x. \label{eq:firstmaxmom}
\end{align}
\end{coroll}
We approximate the expectation of the collateral spread maximum $M(t_k)$ by the expectation of the common factor maximum $\widetilde M(t_k)$,
\begin{equation}
\E\bigl[M(t_k)\bigr] \approx \E\bigl[\widetilde M(t_k)\bigr],
\end{equation}
and repeat these steps at all times $t_k \in \mathcal{T}$. Thus, we arrive at a common factor estimator for the first order approximation in \eqref{eq:exchange-exp}.
%
%
\red{
\begin{dfn}[First order common factor estimator]\label{def:firstordercf}
The first order approximation derived in Section~\ref{sec:approxCTD} of the spread based CTD discount factor in \eqref{eq:ctd-spread} is
\begin{equation}
\exp\Bigl(\E\Bigl[{-\int_0^T  M(t)\d t}\Bigr]\Bigr) \approx \E\Bigl[\exp\Bigl({-\int_0^T  M(t)\d t}\Bigr)\Bigr],
\end{equation}
where $M(t) = \max\bigl(0, q_1(t), \dots, q_{N}(t)\bigr)$ is the maximum of the collateral spreads and zero.
 
 If at every time $t_k$ in a time discretization $\mathcal{T}$ of $[0,T]$, the collateral spreads $q_1(t_k), \dots, q_N(t_k)$ are approximated by the conditionally independent common factor approximation $\widetilde q_1(t_k), \dots, \widetilde q_N(t_k)$ given in Definition~\ref{def:cf}, then the first order common factor estimator is given by 
\begin{equation}\label{eq:CF1}
\mathrm{CF}_1(T) = \exp \Bigl( -\int\limits_{t_k\in\mathcal{T}} \E[\widetilde M(t_k)] \d t_k\Bigr),
\end{equation}
where $\E[\widetilde M(t_k)] = \E[\max(0, \widetilde q_1(t_k), \dots, \widetilde q_N(t_k))]$ is the expectation of the common factor maximum given in Corollary~\ref{cor:expMtilde}.
Here the integral over discrete times $t_k$ indicates the use of a suitable discretization.
\end{dfn}
}
 In the implementation of Section~\ref{sec:example}, we will use a simple discretization in terms of left sums:
\begin{equation}
\mathrm{CF}_1(T)  \approx \exp \Bigl( -\sum\limits_{t_k\in\mathcal{T}} \E[\widetilde M(t_k)] \Delta t_k\Bigr).
\end{equation}
%
%
\section{Second-order models with the common factor approximation}\label{sec:Second-order}
Having obtained the first order common factor estimator \eqref{eq:CF1}, the second-order approximation in \eqref{eq:2ndtaylor},
\begin{equation}
\E\bigl[\exp\bigl(Y(T)\bigr)\bigr] \approx \exp\Bigl(\E\bigl[Y(T)\bigr]\Bigr) \Bigl( 1 + \frac12 \Var\bigl[Y(T)\bigr]\Bigr),
\end{equation}
additionally requires the variance of the integral, $\Var[Y(T)] = \Var[\int_0^T M(t)\d t]$. 
 \par
An explicit solution of this variance depends on covariance terms across times, 
\begin{equation}
\Cov[M(t), M(s)],\quad s,t\in(0,T], 
\end{equation} 
which are not available in closed form for the maximum of Hull--White processes. We will therefore introduce two different approximations of the variance $\Var[Y(T)]$.
\par
The first estimator, called diffusion-based estimator, builds on a model which neglects the mean reverting dynamics of the collateral spreads and approximates the variance of the integral of the maximum with the variance of the integral of a related It\^o process, an approach which turns out to be numerically very efficient. 
\par
The second, mean-reversion-based estimator builds on the approximative variance of the integral expression developed in \cite{SankovichZhu} and accounts for the mean reverting structure of the involved processes. Difficulties arise because the speed of mean reversion of the maximum process $M(t)$ is a stochastic quantity that depends on the component which is largest at the considered time. To solve this, an average speed of mean reversion  $\kappa(t)$ is attributed to the maximum process, which weights the different speed of mean reversion parameters of the components with their probability of being the maximum, 
\begin{equation}
\kappa(t) = \sum_{i=1}^N \P\bigl[q_i(t) = \max(0, q_1(t), \dots, q_N(t))\bigr] \kappa_i.
\end{equation}
Within the framework of the common factor approximation, these probabilities become analytically tractable at interpolation times $t_k \in \mathcal{T}$.
The mean-reversion-based estimator achieves higher accuracy, in particular under large speed of mean reversion parameters, at the expense of additional computational complexity. Accuracy of these two approaches under different model parameters is studied in Section~\ref{sec:example}.
%
%
\subsection{The diffusion-based estimator}\label{subsec:diffusion}
The diffusion-based estimator approximates the variance of the integral of the maximum with the variance of the integral of a standard It\^o process $(X(t))_{t\in[0,T]}$:
\begin{equation}
\Var\bigl[\int_0^T M(t) \d t\bigr] \approx \Var\bigl[\int_0^T X(t)\d t\bigr].
\end{equation}
The process $X$ is defined such that it approximates the variance of the maximum process $M$. This is achieved by matching the marginal distributions of $X(t_k)$ to the variance of the common factor maximum $\widetilde{M}(t_k)$ for all times $t_k\in\mathcal{T}$. 
That is, we require $\Var[X(t_k)] = \Var[\widetilde M(t_k)]$. 
\begin{dfn}[Auxiliary It\^o process]\label{def:ItoX}
Define 
\begin{equation}\label{eq:defX}
\d X(t) = {h(t)}\d W^X(t),\ t\in(0,T],\ X(0)=0,
\end{equation}
with driving Brownian motion $W^X$ independent of the Brownian motions $W_i$, $i\in\{1,\dots,N\}$. 
The volatility coefficient $h$ is defined piecewise constant between interpolation points $t_k\in\mathcal{T}$ by 
\begin{equation}\label{eq:volacoeff}
h^2(t) = \frac{1}{\Delta t}\Bigl(\Var\bigl[\widetilde M(t_{k+1})\bigr] - \Var\bigl[\widetilde M(t_k)\bigr]\Bigr), \ t\in(t_k, t_{k+1}].
\end{equation}
\end{dfn}
As a consequence of \eqref{eq:volacoeff}, at all interpolation points $t_k \in \mathcal{T}$ the desired property holds:
\begin{align}
\Var[X(t_k)] &= \int_0^{t_k} h^2(s)\d s = \Var[\widetilde M(t_k)]. \label{eq:hasM}
\end{align}
The next lemma expresses the integral variance $\Var[\int_0^T X(t)\d t]$ as a function of the volatility coefficient $h$.
\begin{lem}[Integral variance of $X(t)$]\label{lem:intvar}
Let $X(t)$ be the It\^o process given in Definition~\ref{def:ItoX}. The variance of its integral is given by
\begin{align}\label{eq:intvarX}
\Var\Bigl[\int_0^T X(t) \d t\Bigr] &= 
\int\limits_0^T\int\limits_0^s \Bigl(\int\limits_0^t h^2(u)\d u\Bigr) \d t \d s + \int\limits_0^T (T-s) \Bigl(\int\limits_0^s h^2(t)\d t\Bigr)\d s.
\end{align}
\end{lem}
\begin{proof}
This result follows directly from the solution $X(t) = \int_0^t {h(u)}\d W^X(u)$. By It\^o isometry it holds that
\begin{align}
\int\limits_0^T\int\limits_0^T \E\Bigl[\int_0^t {h(u)}\d W^X(u) \int_0^s {h(v)}\d W^X(v)\Bigr] \d t \d s &= 
\int\limits_0^T \int\limits_0^T \int\limits_0^{t\wedge s} h^2(u) \d u \d t \d s \nonumber \\
&= \int\limits_0^T\int\limits_0^s \Bigl(\int\limits_0^t h^2(u)\d u\Bigr) \d t \d s \nonumber \\
&\hspace{1em} + \int\limits_0^T\int\limits_s^T \Bigl(\int\limits_0^s h^2(u)\d u\Bigr) \d t \d s.
\end{align}
The second term in the sum equals
\begin{equation}
\int\limits_0^T\int\limits_s^T \Bigl(\int\limits_0^s h^2(u)\d u\Bigr) \d t \d s = \int\limits_0^T(T-s)\Bigl(\int\limits_0^s h^2(u)\d u\Bigr) \d s,
\end{equation}
which concludes the proof.
\end{proof}
In the numerical evaluation over the time discretization $\mathcal{T}$, the formula for the variance of the integral of $X$ can be further simplified by replacing the inner integrals with \eqref{eq:hasM}.
\par
It follows from Lemma~\ref{lem:intvar}, that the variance of the integral of $X$ depends only on the volatility coefficient $h$, which is determined by the variance of the common factor maximum, $\Var[\widetilde M(t_k)]$ at the interpolation times $t_k\in\mathcal{T}$.
In the following lemma, we show how this variance can be computed. This result is comparable to Corollary~\ref{cor:expMtilde} in the first order approximation, which gave the expectation of the common factor maximum.
\par
First, we introduce additional notation. Recall that under the common factor approximation of Definition~\ref{def:cf}, the common factor $C(t_k)$ is a centred normal random variable, whose density function we denote by
\begin{equation}\label{eq:pdfC}
f_{C(t_k)}(x) = \frac{1}{\sqrt{2\pi\gamma(t_k)\sigma^2_{\min}(t_k)}}\exp\Bigl(-\frac{x}{\sqrt{\gamma(t_k)\sigma^2_{\min}(t_k)}}\Bigr),\ x\in\R.
\end{equation}
We further denote the cumulative distribution function of the maximum of the individual random variables, $\max(A_1(t_k), \dots, A_N(t_k))$ by $F_{\max_i(A_i(t_k))}$. With the help of Lemma~\ref{lem:max-indep} it is found to be
\begin{equation}\label{eq:cdfmaxAi}
F_{\max_i(A_i(t_k))}(x) = \prod\limits_{i=1}^N \Phi\Bigl(\frac{x-\mu_i(t_k)}{\sqrt{\sigma_i^2(t_k) - \sigma_{\min}^2(t_k)|\gamma(t_k)|}}\Bigr),\ x\in\R.
\end{equation}
These two functions in place, we can give the variance of the common factor maximum.
\begin{lem}[Variance of the common factor maximum]\label{lem:cfmaxvar}
Let $\widetilde M(t_k) = \max(0, \widetilde q_1(t_k), \dots, \widetilde q_N(t_k))$ be the common factor maximum at time $t_k\in\mathcal{T}$.
The variance of $\widetilde M(t_k)$ is given by
\begin{equation}\label{eq:convolutionvar}
\begin{aligned}
\Var\bigl[\widetilde M(t_k)\bigr] = \int_0^\infty 2x\Bigl(1- \bigl( f_{C(t_k)} * F_{\max_i(A_i(t_k))} \bigr)(x) \Bigr)\d x \\
 - \left( \int_0^\infty 1 -  \bigl( f_{C(t_k)} * F_{\max_i(A_i(t_k))} \bigr) (x) \d x \right)^2,
\end{aligned}
\end{equation}
where $*$ denotes the convolution operator.
\end{lem}
\begin{proof}Let $t_k\in\mathcal{T}$ be fixed. The common factor maximum $\widetilde M(t_k)$ can be expressed as
\begin{align}
\widetilde M(t_k) = \max(0, \widetilde q_1(t_k), \dots, \widetilde q_N(t_k)) &= \max(0, \max(\widetilde  q_1(t_k), \dots, \widetilde q_N(t_k))) \nonumber \\ 
&= \max(0, C(t_k) + \max(A_1(t_k), \dots, A_N(t_k))). \label{cfmaxothertrick}
\end{align}
Let $G_k$ denote the cumulative distribution function of the random variable $ \max(0, C(t_k) + \max(A_1(t_k), \dots, A_N(t_k)))$ and let $H_k$ denote the cumulative distribution function of the random variable $C(t_k) + \max(A_1(t_k), \dots, A_N(t_k))$.
Then,
\begin{equation}
G_k(x) = 
	\begin{cases}
	0, & x < 0,\\
	H_k(x), & x \geq 0.
	\end{cases}
\end{equation}
As it holds $\Var[\widetilde M(t_k)] = \E[\widetilde M(t_k)^2] - \E[\widetilde M(t_k)]^2$, it suffices to obtain the first two moments of $\widetilde M(t_k)$ from the cumulative distribution function with Lemma~\ref{lem:cdfmoments}. These are given by
\begin{align}
\E[\widetilde M(t_k)^2] &= \int_{-\infty}^0 2x G_k(x) \d x + \int_0^\infty 2x\bigl(1-G_k(x)\bigr)\d x \nonumber \\
&= \int_0^\infty 2x\bigl(1-H_k(x)\bigr)\d x. \label{eq:newVar}
\end{align}
and
\begin{align}
\E[\widetilde M(t_k)] &= \int_{-\infty}^0 - G_k(x) \d x + \int_0^\infty \bigl(1-G_k(x)\bigr)\d x \nonumber \\
& = \int_0^\infty 1 -  H_k (x) \d x. \label{eq:convolutionexpec}
\end{align}

It remains to show that $H_k(x) = (f_{C(t_k)} * F_{\max_i(A_i(t_k))})(x)$. As $C(t_k)$ and $A_i(t_k)$ are independent for all $i$, $C(t_k)$ is also independent of $\max(A_1(t_k), \dots, A_N(t_k))$ and thus it holds that
\begin{align}
H_k(x) &= \P[C(t_k) + \max(A_1(t_k), \dots, A_N(t_k)) \leq x] \nonumber \\
&= \int_{-\infty}^\infty \int_{-\infty}^{x-z} f_{C(t_k), \max_i(A_i(t_k))}(z,y)\d y \d z \nonumber \\
&= \int_{-\infty}^\infty \int_{-\infty}^{x-z} f_{C(t_k)}(z) f_{\max_i(A_i(t_k))}(y)\d y \d z \nonumber \\
&= \int_{-\infty}^\infty f_{C(t_k)}(z) F_{\max_i(A_i(t_k))}(x-z) \d z \nonumber \\
&= \bigl( f_{C(t_k)} * F_{\max_i(A_i(t_k))} \bigr) (x).
\end{align}
Here, $f_{C(t_k), \max_i(A_i(t_k))}$ denotes the joint density of $C(t_k)$ and $\max(A_1(t_k), \dots, A_N(t_k))$ which factors into the marginal densities $f_{C(t_k)}$ and $f_{\max_i(A_i(t_k))}$ by independence.
\end{proof}
\par
\medskip
We call the variance of the integral of the auxiliary It\^o process the {diffusion-based variance estimator} of the variance of the integral of the maximum of the collateral spreads, and denote it by $\Psi(T)$:
\begin{equation}\label{eq:diffvarestimator}
\Psi(T) := \Var[\int_0^T X(t)\d t] \approx \Var[\int_0^T M(t)\d t].
\end{equation}
This allows us to define a first, diffusion-based CTD estimator of the second-order approximation in \eqref{eq:2ndtaylor}.
%
%
\red{
\begin{dfn}[The second-order common factor estimator 1]\label{def:CF2-diff}
The second-order approximation derived in Section~\ref{sec:approxCTD} of the spread based CTD discount factor in \eqref{eq:ctd-spread} is
\begin{equation}
\exp\Bigl(\E\Bigl[{-\int_0^T  M(t)\d t}\Bigr]\Bigr) \Bigl(1 + \frac12 \Var\int_0^T  M(t)\d t  \Bigr) \approx \E\Bigl[\exp\Bigl({-\int_0^T  M(t)\d t}\Bigr)\Bigr],
\end{equation}
where $M(t) = \max\bigl(0, q_1(t), \dots, q_{N}(t)\bigr)$ is the maximum of the collateral spreads and zero.

Let the conditional independence assumptions hold as in Definition~\ref{def:firstordercf} and let $\mathrm{CF}_1(T)$ denote the first order common factor estimator,
\begin{equation}
\mathrm{CF}_1(T) = \exp \Bigl( -\int\limits_{t_k\in\mathcal{T}} \E[\widetilde M(t_k)] \d t_k\Bigr),
\end{equation}
where $t_k\in\mathcal{T}$ are times in a suitable time discretization of $[0,T]$.

Assume that the variance of the integral of the maximum process can be approximated by the variance of the integral of the It\^o process $X$, given in Definition~\ref{def:ItoX}, with variance $\Var[X(t_k)] = \Var[\widetilde M(t_k)]$ obtained from the common factor maximum in Lemma~\ref{lem:cfmaxvar}.

Then the second-order common factor estimator with a diffusion-based integral variance estimator is given by 
\begin{align}\label{eq:CF2-diff}
\mathrm{CF}_2^{(1)}(T) &:= \mathrm{CF}_1(T) \Bigl(1 + \frac12\Psi(T)\Bigr),
\end{align}
where $\Psi(T) = \Var[\int_0^T X(t)\d t]$ is the diffusion-based variance estimator given by Lemma~\ref{lem:intvar}.
\end{dfn}
}
\subsection{The mean-reversion-based estimator}\label{subsec:meanreverting}
A more accurate approximation of $\Var[\int_0^T M(t)\d t]$ can be found on the basis of the dynamics of the maximum process $M(t)$. 
In \cite{SankovichZhu}, an estimator is developed for the maximum of Hull--White collateral rates, a setting where no constant zero process is part of the maximum. We briefly outline the argument in our setting of collateral spreads, referring to the aforementioned reference for a detailed proof \red{and details of all technical assumptions}.
\par
For the moment, assume that the spread $q_0$ is also a Hull--White process with speed of mean reversion parameter $\kappa_0 > 0$, in line with the remaining spreads $q_i$, $i\in\{1,\dots, N\}$.
Then, the dynamics of the maximum process $\max(q_0(t), \dots, q_N(t))$ can be obtained by the It\^o-Tanaka formula (see e.g.\ \cite{revuzyor}). However, these dynamics are not of a closed, purely analytical form but instead they contain additional stochastic terms, which create a path dependence to whichever component of the maximum is the maximal component at a time.
As a consequence, the mean reversion speed of the maximum process $\max(q_0(t), \dots, q_N(t))$ takes the form of a random variable:
\begin{equation}\label{eq:stochMR}
\sum_{i=0}^N 1_{\{\max(q_0(t),\dots,q_N(t))=q_i(t)\}}\, \kappa_i.
\end{equation}
Define the event that $q_i$ is the maximal spread at time $t$ as $D_i(t) = \{\max(q_0(t),\dots,q_N(t))=q_i(t)\}$ for all $i\in\{0, \dots, N\}$.
Then analytic tractability is restored by replacing the stochastic expression \eqref{eq:stochMR} with its expectation. This defines a weighted speed of mean reversion $\kappa(t)$ given by
\begin{align}
\kappa(t) = \sum_{i=0}^N \E\bigl[1_{D_i(t)}\bigr]\, \kappa_i 
= \sum_{i=0}^N \P\bigl[D_i(t)\bigr]\, \kappa_i.
\end{align}

With additional simplifying assumptions (we again refer to \cite{SankovichZhu}), an estimator for the variance of the integral of the maximum is found, which takes the form
\begin{multline}\label{eq:longeq}
\Var\bigl[\int\limits_0^T  \max(q_0(t), \dots, q_N(t))\d t\bigr] \approx \\ 
2\int\limits_0^T  \e^{-\int\limits_0^t\kappa(u)\d u} 
\int\limits_0^t \e^{\int\limits_0^s\kappa(u)\d u}\Var\bigl[\max\bigl(q_0(s),\dots,q_N(s)\bigr)\bigr]\d s \d t.
\end{multline}

\par
We return to our setting of the maximum process $M(t)=\max(0, q_1(t), \dots, q_N(t))$. In particular, the first component $q_0 = 0$ is not a mean-reverting process any longer. 
Thus, in the weighted mean reversion function 
\begin{equation}
\kappa(t) = \sum\limits_{i=0}^N\P\bigl[M(t)=q_i(t)\bigr] \kappa_i, \label{eq:kappa(t)}
\end{equation}
there is no natural candidate for what should be the speed of mean reversion $\kappa_0$ of the zero spread. Therefore, we remove this component from the weighted mean reversion function $\kappa(t)$ by setting $\kappa_0 = 0$.
\par
The probabilities $\P[M(t) = q_i(t)]$ in the weighted mean reversion function can be estimated with the common factor approximation by substituting $\P[\widetilde M(t) = \widetilde q_i(t)]$ which is easily computed, as the next lemma shows.
\begin{lem}\label{lem:maxprobs}[Common Factor Maximum Probabilities] 
Under the assumptions of Definition~\ref{def:cf}, for all times $t_k\in\mathcal{T}$ and $i\in\{1,\dots,N\}$, the probability that the common factor spread $\widetilde q_i(t_k)$ equals the common factor maximum $\widetilde M(t_k)$ is given by
\begin{equation}\label{eq:maxprobs}
\P\bigl[\widetilde M(t_k) = \widetilde q_i(t_k)\bigr] = \int\limits_{\R} f_{A_i(t_k)}(x) F_{C(t_k)}(x)\prod\limits_{\substack{j=1\\ j\neq i}}^N F_{A_j(t_k)}(x) \d x,
\end{equation}
where $f_{A_i(t_k)}$ denotes the density of the normal random variable $A_i(t_k)$ and $F_{C(t_k)}$, $F_{A_j(t_k)}$ the cumulative distribution functions of the normal random variables $C(t_k)$ and $A_j(t_k)$, respectively. 
\end{lem}
\begin{proof}
By definition of the common factor approximation, it holds
\begin{align}
\P\bigl[\widetilde M(t_k) = \widetilde q_i(t_k)\bigr] &= \P\bigl[C(t_k) + \max\bigl(-C(t_k), A_1(t_k), \dots, A_N(t_k)\bigr) = C(t_k) + A_i(t_k)\bigr] \nonumber \\
&= \P\bigl[\max\bigl(-C(t_k), A_1(t_k), \dots, A_N(t_k)\bigr) = A_i(t_k)\bigr] \nonumber \\
&= \P\bigl[C(t_k)\leq A_i(t_k), A_1(t_k)\leq A_i(t_k), \dots, A_N(t_k)\leq A_i(t_k)\bigr], \label{eq:Aievent}
\end{align}
where we used equal distributions $C(t_k) \overset{d}= -C(t_k)$ in the final step, as $C(t_k)$ is a centred normal random variable.
The events in \eqref{eq:Aievent} are independent from another when conditioned on $A_i(t_k)$, and all involved random variables are independent from another. Thus it follows
\begin{align}
\P\bigl[\widetilde M(t_k) = \widetilde q_i(t_k)\bigr] &= \E\Bigl[\E\bigl[ 1_{\{C(t_k)\leq A_i(t_k)\} \cap \{A_1(t_k)\leq A_i(t_k)\}\cap \dots \cap \{A_1(t_k)\leq A_i(t_k)\}} \mid A_i(t_k)\bigr]\Bigr] \nonumber \\ 
&=\int\limits_{\R}  \P\bigl[C(t_k)\leq x\mid A_i(t_k) = x\bigr] \prod\limits_{\substack{j=1\\ j \neq i}}^N \P\bigl[A_j(t_k)\leq x\mid A_i(t_k) = x\bigr]  f_{A_i(t_k)}(x)\d x \nonumber \\
&= \int\limits_{\R} f_{A_i(t_k)}(x) F_{C(t_k)}(x)  \prod\limits_{\substack{j=1\\ j \neq i}}^N F_{A_j(t_k)}(x)\d x,
\end{align}
which finishes the proof.
\end{proof}
We can thus define the weighted mean reversion function $\widetilde \kappa(t)$ of the common factor maximum, given by
\begin{equation}
\widetilde \kappa(t) = \sum_{i=1}^N \P[\widetilde M(t) = \widetilde q_i(t)] \kappa_i.
\end{equation}
We again estimate the variance of the maximum $\Var[M(t_k)]$ with the variance of the common factor maximum, given in Lemma~\ref{lem:cfmaxvar}. Then, the mean-reversion-based estimator for the variance of the integral of the maximum of the collateral spreads, which we denote by $\chi(T)$, is obtained:
\begin{align}\label{eq:sanko-variance}
\chi(T) := 2\int_0^T \e^{-\int_0^t\widetilde \kappa(u)\d u}\int_0^t\left(\e^{\int_0^s \widetilde\kappa(u)\d u}\Var\bigl[\widetilde M(s)\bigr]\right)\d s\, \d t \approx \Var\bigl[\int_0^T M(t) \d t\bigr] .
\end{align}
This allows us to define a second, mean-reversion-based CTD estimator for the second-order approximation in \eqref{eq:2ndtaylor}.
%
\red{
\begin{dfn}[The second-order common factor estimator 2]\label{def:CF2-MR}
The second-order approximation derived in Section~\ref{sec:approxCTD} of the spread based CTD discount factor in \eqref{eq:ctd-spread} is
\begin{equation}
\exp\Bigl(\E\Bigl[{-\int_0^T  M(t)\d t}\Bigr]\Bigr) \Bigl(1 + \frac12 \Var\int_0^T  M(t)\d t  \Bigr) \approx \E\Bigl[\exp\Bigl({-\int_0^T  M(t)\d t}\Bigr)\Bigr],
\end{equation}
where $M(t) = \max\bigl(0, q_1(t), \dots, q_{N}(t)\bigr)$ is the maximum of the collateral spreads and zero.

Let the conditional independence assumptions hold as in Definition~\ref{def:firstordercf} and let $\mathrm{CF}_1(T)$ denote the first order common factor estimator,
\begin{equation}
\mathrm{CF}_1(T) = \exp \Bigl( -\int\limits_{t_k\in\mathcal{T}} \E[\widetilde M(t_k)] \d t_k\Bigr),
\end{equation}
where $t_k\in\mathcal{T}$ are times in a suitable time discretization of $[0,T]$.

Assume at every time $t_k$ that the probability of the common factor estimates being the common factor maximum approximates the probability of each spread being the maximal spread,
\begin{equation}
\P\bigl[\widetilde M(t_k) = \widetilde q_i(t_k)\bigr] \approx \P\bigl[M(t_k) = q_i(t_k)\bigr], \quad i\in\{1, \dots, N\},
\end{equation}
and assume further that the variance of the maximum of the spreads can be approximated by the variance of the common factor maximum, $\Var[M(t_k)] = \Var[\widetilde M(t_k)]$.

Then the second-order common factor estimator with a mean-reversion-based integral variance estimator is given by 
\begin{equation}\label{eq:CF2-mr}
\mathrm{CF}_2^{(2)}(T) := \mathrm{CF}_1(T) \Bigl(1 + \frac12\chi(T)\Bigr),
\end{equation}
where $\chi(T) \approx \Var[\int_0^T M(t)\d t]$ is the mean-reversion-based variance estimator given by \eqref{eq:sanko-variance}.
\end{dfn}
}
\medskip
In summary, we obtain two different estimators for the second-order approximation of the CTD. The first, $\mathrm{CF}_2^{(1)}(t)$ defined in Definition~\ref{def:CF2-diff}, bases the variance of the integral on diffusion dynamics, and the second, $\mathrm{CF}_2^{(2)}(T)$ defined in Definition~\ref{def:CF2-MR}, bases the variance of the integral on mean-reverting dynamics.
%
\section{Model extension}\label{sec:extension}
As mentioned above, so far the common factor approximation introduced in Section~\ref{subsec:cf} is limited in the correlation structures it can model, due to the critical step in \eqref{eq:indep-max}, 
\begin{equation}
\max\bigl(0, C(t_k) + A_1(t_k), \dots, C(t_k) + A_N(t_k)\bigr) = C(t_k) + \max\bigl(-C(t_k), A_1(t_k), \dots, A_N(t_k)\bigr),
\end{equation}
and the equivalent step in \eqref{cfmaxothertrick}.
All common factor approximations $\widetilde q_i(t_k)$ need to contain the common factor $C(t_k)$ with the same sign and magnitude, so that $C(t_k)$ can be fully subtracted from each component. Consequently, as the common factor determines the correlation structure, the components will always be positively correlated. 
This issue can be resolved by extending the model to multiple subcategories with category-specific common factors $C^{(1)}(t_k)$, $C^{(2)}(t_k)$, and so on. By introducing correlations between these category-specific common factors, various correlation structures emerge in the common factor approximation. This is shown in detail in Appendix~\ref{sec:appx1}.
%
%
%
\section{Numerical results} \label{sec:example}
%
In this section, we present insights into the numerical computation of the common factor approximation models and further analyse their sensitivity to changes in the input parameters. This consideration is done for a three currency setup, where the parameters of the collateral spread model, defined in Section~\ref{sec:coll-spread-model}, are given in Table~\ref{table:spreadparams}. Note that only two collateral spreads need to be defined, as the third currency is embedded in the constant zero spread. The collateral spread parameters used here are derived from the collateral rate parameters given in \cite{SankovichZhu}, where also calibration is discussed. The exact details of this conversion from rates to spreads are given in Appendix~\ref{sec:appxparams}.  For each spread, the long-term mean levels $\theta_i$ are \red{chosen such that the expectations of the spreads equal their initial value.}
\par
The algorithm for computing the first and second-order common factor estimators, initialized with the Hull--White parametrization of the spreads $q_i$, is outlined in Appendix~\ref{sec:appxalgo}. 
\par
\begin{table}[]
\caption{Spread parameters}\label{table:spreadparams}
\centering
\begin{tabular}{ll|ll|ll}
$\kappa_1$ & $0.0078$ & $\xi_1$ & $0.0018$ & $q_1(0)$ & $0.000845$ \\ 
$\kappa_2$ & $0.0076$ & $\xi_2$ & $0.0023$ & $q_2(0)$ & $0.001514$ 
\end{tabular}
\end{table}
\par
\smallskip
Before discussing the accuracy results of the common factor CTD estimators, we first focus on some specific points in the numerical implementation. At each interpolation time $t_k\in\mathcal{T}$, both moments of the common factor maximum, $\widetilde M(t_k) = \max(0, \widetilde q_1(t_k), \dots, \widetilde q_N(t_k))$, depend on the same continuous convolution $f_{C(t_k)} * F_{\max_i(A_i(t_k))}$ between the density of the common factor and the cumulative distribution function of the maximum of  individual factors, defined in \eqref{eq:pdfC} and \eqref{eq:cdfmaxAi}, respectively. That is, at each time $t_k$, we need to compute the integrals
\begin{align}
\E[\widetilde M(t_k)] &= \int_0^\infty 1 -  \bigl( f_{C(t_k)} * F_{\max_i(A_i(t_k))} \bigr) (x) \d x, \label{eq:repconvolutionexpec}\\
\Var\bigl[\widetilde M(t_k)\bigr] &= \int_0^\infty 2x\Bigl(1- \bigl( f_{C(t_k)} * F_{\max_i(A_i(t_k))} \bigr)(x) \Bigr)\d x \nonumber \\
 &- \left( \int_0^\infty 1 -  \bigl( f_{C(t_k)} * F_{\max_i(A_i(t_k))} \bigr) (x) \d x \right)^2, \label{eq:repconvolutionvar}
\end{align}
derived in \eqref{eq:convolutionexpec} and \eqref{eq:convolutionvar}. 
These indefinite integrals appear computationally demanding, but can be efficiently implemented: The convolution  $(f_{C(t_k)} * F_{\max_i(A_i(t_k))})(x)$ quickly converges to $1$ as $x$ increases, which makes the integrands of \eqref{eq:repconvolutionexpec} and \eqref{eq:repconvolutionvar}  numerically distinguishable from zero only on a small domain $[0, L(t_k)]$ for some $L(t_k)\in\R$. Then, given a discretization of $[0, L(t_k)]$, the convolution can be evaluated efficiently  with the Fast Fourier Transform algorithm (\cite{CooleyTukey}). 
\par
In Table~\ref{table:moments} we assess an appropriate step size $\delta$ for the discretization of $[0, L(t_k)]$,  under the parameters of Table~\ref{table:spreadparams} with instantaneous correlation $\rho_{1,2}=0.3$. Note that the width of the domain, $L(t_k)$, increases with time as the variances of the collateral spreads increase. 
We obtain precise results for the integral moments $\E[\int_0^T M(t)\d t]$ and $\Var[\int_0^T M(t)\d t]$ by Monte Carlo simulation and compare these to the common factor estimators for the first and second integral moment, 
\begin{align}
 \sum_{t_k\in\mathcal{T}} \E[\widetilde M(t_k)] \Delta t_k \approx \E\Bigl[\int_0^T M(t)\d t\Bigr], \\
\Psi(T) \approx \Var\Bigl[\int_0^T M(t)\d t\Bigr],
\end{align}
where $\Psi(T)$ is the diffusion-based integral variance estimator given in \eqref{eq:diffvarestimator}. In both computations we consider a time discretization of step size $\Delta t_k = 0.1$. We find that  $\delta = 5 \times 10^{-5}$ in the convolution grid offers a good compromise between computational speed and precision, and use this grid in the convolution for the remainder of the section.

\begin{table}[]
\caption{In the computation of the expectation and variance of the common factor maximum $\widetilde M(t_k)$ at interpolation times $t_k$, a continuous convolution needs to be evaluated numerically. For different step sizes $\delta$ of the discretizing grid and maturities $T$, we give the relation between this step size and the resulting first and second-order estimators of the maximum integral $\int_0^T \max(0, q_1(t), q_2(t))\d t$.
}\label{table:moments}
\begin{tabular}{llllll}
\textbf{}                   & \textbf{}                      & \textbf{T=5} & \textbf{T=10} & \textbf{T=15} & \textbf{T=20} \\ \hline 
                            & No.\  points in $[0,L(T)]$                      & 83           & 117           & 140           & 158           \\
$\delta = 5 \times 10^{-4}$ & Error of expectation integral   & 0.001052     & 0.002049      & 0.003011      & 0.004003      \\
                            & Error of estimator $\Psi(T)$   & 0.000016     & 0.00011       & 0.000354      & 0.000831      \\ \hline
                            & No.\  points in $[0,L(T)]$                      & 818          & 1109          & 1318          & 1483          \\
$\delta = 5 \times 10^{-5}$ & Error of expectation integral   & 0.000106     & 0.000194      & 0.000295      & 0.000429      \\
                            & Error of estimator $\Psi(T)$   & 0.000012     & 0.000086      & 0.000287      & 0.000691      \\ \hline
                            & No.\  points in $[0,L(T)]$ & 7563         & 10442         & 12490         & 14219         \\
$\delta = 5 \times 10^{-6}$ & Error of expectation integral   & 0.000012     & 0.000012      & 0.000032      & 0.000064      \\
                            & Error of estimator $\Psi(T)$   & 0.000012     & 0.000084      & 0.000279      & 0.000675      \\ \hline
\end{tabular}
\end{table}

\subsection{Model accuracy}\label{subsec:sensi}
We remain in the three currency setup of Table~\ref{table:spreadparams} and consider the model accuracy for different values of the characterizing Hull--White parameters. Throughout this section, a maturity of $T=20$ years and a time discretization of step size $\Delta t_k = 0.1$ remains fixed.
For reference values, we use Monte Carlo simulation over the same time discretization to compute precise estimates of the approximated object of interest, given by
\begin{equation}
\E^T\Bigl[\exp\Bigl(-\int_0^T \max\bigl(0, q_1(t), q_2(t)\bigr)\d t\Bigr)\Bigr]
\end{equation} 
in \eqref{eq:ctd-spread}.
These reference values are compared to the first order common factor estimator $\mathrm{CF}_1(T)$, the second-order common factor estimator with diffusion-based variance $\mathrm{CF}_2^{(1)}(T)$ and the second-order common factor estimator with mean-reversion-based variance $\mathrm{CF}_2^{(2)}(T)$, defined in (\ref{eq:CF1}), (\ref{eq:CF2-diff}) and (\ref{eq:CF2-mr}), respectively.
\par
We begin by considering different values for the instantaneous correlation parameter $\rho_{1,2}$. By Lemma~\ref{lem:cfprops} part 3, the correlation of common factor approximations $\corr(\widetilde q_1(t_k),\widetilde q_2(t_k))$ is bounded, thus there is a maximal value of the instantaneous correlation parameter, above which the correlation cannot be expressed by the common factor model any longer because the correlation optimization parameter  $\gamma(t_k)$ has reached its maximum of $1$. For our chosen parameters, this boundary is at $\rho_{1,2} = 0.78$, as shown in Figure~\ref{fig:corr-sensi}. As the error graph to the right shows, the quality of the second-order approximations is uniformly very good for admissible correlation values.
\begin{figure}[]
\centering
\includegraphics[width=.49\textwidth, height=0.40833\textwidth]{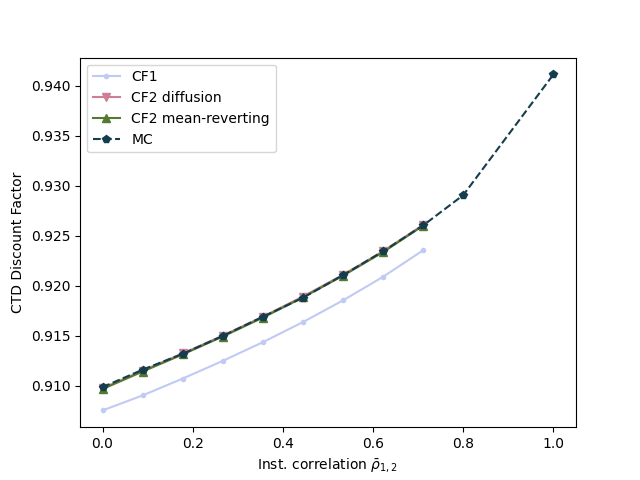}
\includegraphics[width=.49\textwidth, height=0.40833\textwidth]{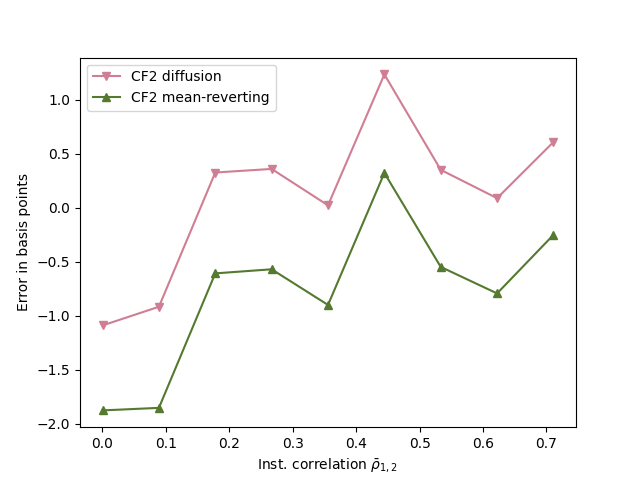}
\caption{Common factor approximations for three currencies plotted against variation in the instantaneous correlation parameter $\rho_{1,2}$. The collateral spreads are modelled with the parameters in Table~\ref{table:spreadparams}.}
\label{fig:corr-sensi}
\end{figure}
\par
We turn to the effect of the speed of mean reversion parameter. In Figure~\ref{fig:mr-sensi-log}, an instantaneous correlation value of $\rho_{1,2}=0.3$ is fixed and the speed of mean reversion parameters $\kappa_1$ and $\kappa_2$ is varied. To obtain a two-dimensional graph, both parameters are scaled with the same factor and their average,  $(\kappa_1 + \kappa_2)/2$, is drawn on the horizontal axes.
The error graph demonstrates the advantage of the mean-reversion-based second-order estimator, which is very accurate even under high speed of mean reversion. However, this comes at the cost of a higher numerical complexity which will be demonstrated in Section~\ref{subsec:timing}.
\begin{figure}[]
\centering
\includegraphics[width=.49\textwidth, height=0.40833\textwidth]{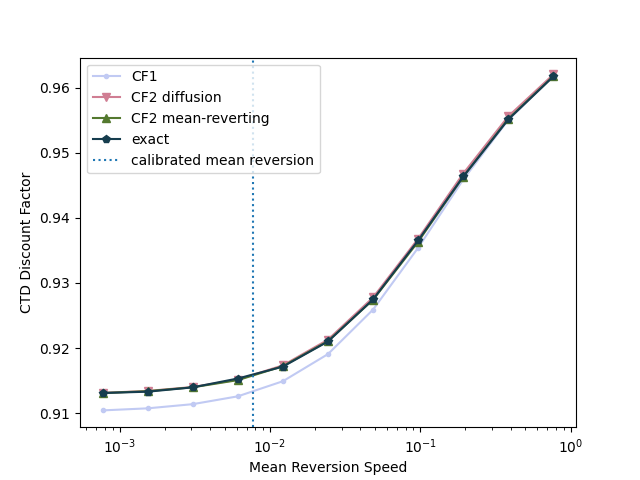}
\includegraphics[width=.49\textwidth, height=0.40833\textwidth]{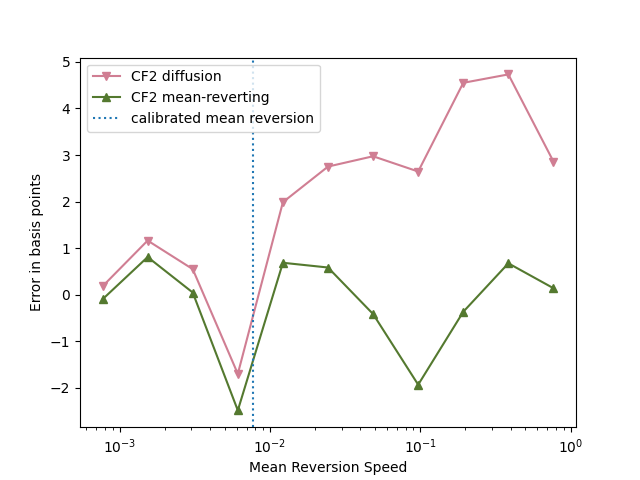}
\caption{Common factor approximations for three currencies plotted against variation in the mean reversion parameters. The mean reversion parameters $\kappa_1$ and $\kappa_2$ are varied at the same rate, the horizontal axes show the average speed of mean reversion $(\kappa_1 + \kappa_2)/2$ on a logarithmic scale.}
\label{fig:mr-sensi-log}
\end{figure}
\par
\red{As discussed in Section~\ref{sec:introduction}, collateral spreads exhibit less volatility than their associated, highly correlated collateral rates. This results in the maximum distribution of the spreads having less skewness than the maximum distribution of the rates, and therefore their distribution can be accurately approximated with a second-order model.}
This is illustrated in  Figure~\ref{fig:vola-sensi}, which shows that increasing the volatility back to levels expressed in the collateral rate formulation (where volatility parameters are around $0.007$), drastically increases the error of the second-order approximations.
\begin{figure}[]
\centering
  \includegraphics[width=.49\textwidth, height=0.40833\textwidth]{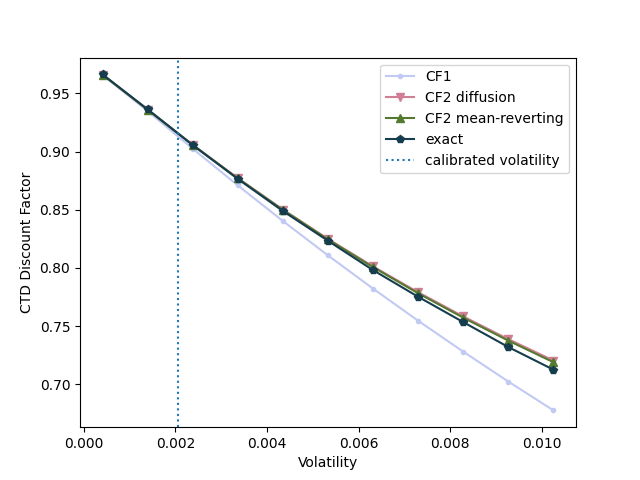}
  \includegraphics[width=.49\textwidth, height=0.40833\textwidth]{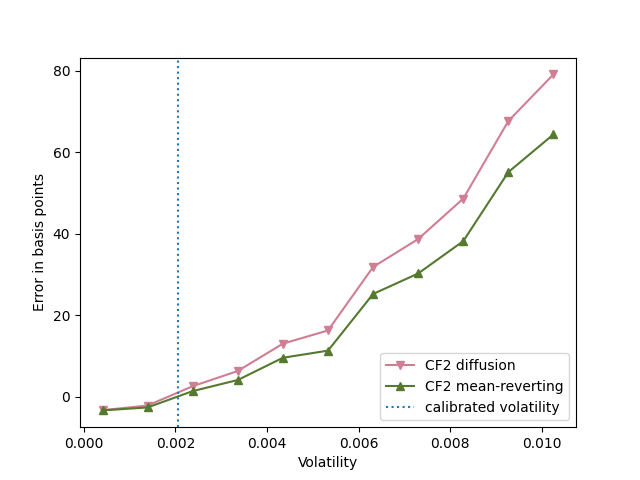}
\caption{Common factor approximations for three currencies plotted against variation in the volatility parameters. The volatility parameters $\xi_1$ and $\xi_2$ are varied at the same rate, the horizontal axes show the average volatility parameter $(\xi_1 + \xi_2)/2$.}
\label{fig:vola-sensi}
\end{figure}
%
\red{
\subsection{Robustness in a stressed market}
In the spring of 2020, in the wake of the COVID-19 pandemic, cross-currency spreads widened, which registered in the collateral spread projections, depicted in Figure~\ref{fig:stressed}. Compared to previous projections, the spreads both widened in the short and medium term and became more ambiguous in the long term. 
Given this term structure, Figure~\ref{fig:stressed} demonstrates robustness of the common factor approximation in a stressed scenario. The error over different maturities increases mildly under the parameters of Table~\ref{table:spreadparams}. Even doubling the volatility coefficient to reflect a surge in market turbulence only results in a 20-year difference of 9, respectively 12, basis points for the mean-reverting-based and diffusion-based second-order estimators. The deterministic approach is not included in the graphical representation, as the difference to the exact result is well above 5\%.
\begin{figure}[]
\centering
  \includegraphics[width=.49\textwidth, height=0.40833\textwidth]{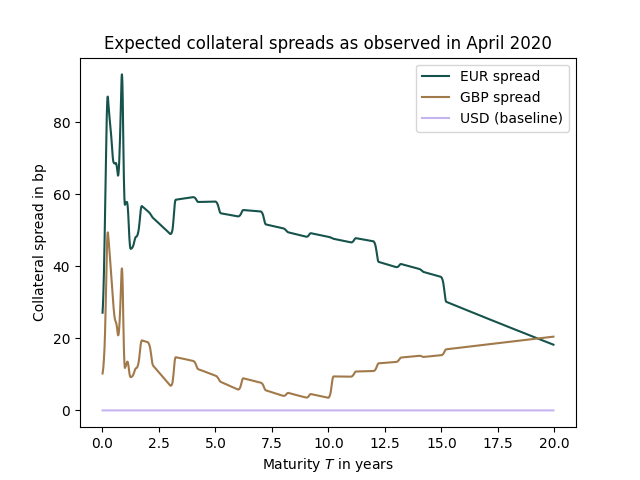}
  \includegraphics[width=.49\textwidth, height=0.40833\textwidth]{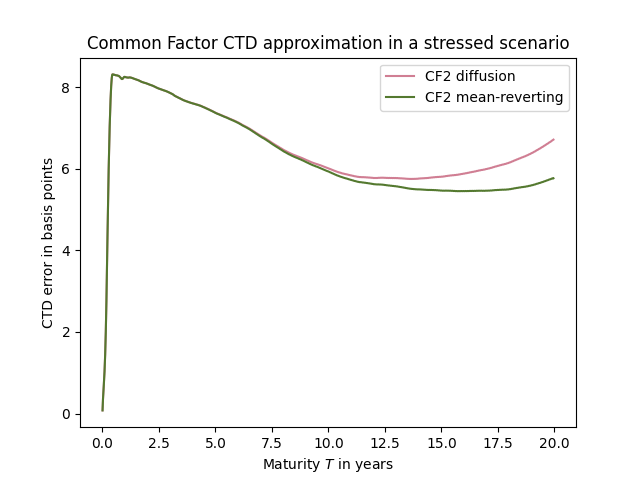}
\caption{\red{Stressed collateral spreads under market turbulence as observed in April 2020. Long term mean dynamics are displayed on the left, the resulting error over time of the common factor second-order approximations with these means 
 is given on the right.}}
\label{fig:stressed}
\end{figure}
}
\subsection{Computation time}\label{subsec:timing}
\begin{table}[]
\centering
\caption{Relative increase in computation time of the second-order estimators over the base-case of $3$ currencies. 
For each number of currencies, the proportion of time spent on computing the moments of the common factor maximum $\widetilde M(t)$ from the convolution integrals is given.
}
\label{table:speed}
\begin{tabular}{ccccc}
No.\ Currencies & 
\begin{tabular}{@{}c@{}}Relative time \\ $\mathrm{CF}_2^{(1)}(20)$\end{tabular} & 
\begin{tabular}{@{}c@{}}Relative time \\ $\mathrm{CF}_2^{(2)}(20)$\end{tabular} & 
\begin{tabular}{@{}c@{}}Moment proportion \\ $\mathrm{CF}_2^{(1)}(20)$\end{tabular} & 
\begin{tabular}{@{}c@{}}Moment proportion \\ $\mathrm{CF}_2^{(2)}(20)$\end{tabular} \\ \hline
\textbf{3}                      & 1                                                                                                                        & 1                                                                                                                        & 99.2\%                         & 43.9\%                         \\
\textbf{4}                      & 1.38                                                                                                                    & 1.50                                                                                                                    & 93.1\%                         & 37.7\%                         \\
\textbf{5}                      & 1.80                                                                                                                  & 2.08                                                                                                                   & 94.3\%                         & 35.9\%                         \\
\textbf{6}                      & 2.11                                                                                                                 & 2.68                                                                                                                   & 95.0\%                         & 32.9\%                         \\
\textbf{7}                      & 2.40                                                                                                                   & 3.31                                                                                                                   & 95.5\%                         & 30.6\%                         \\
\textbf{8}                      & 2.79                                                                                                                   & 3.99                                                                                                                   & 95.9\%                         & 29.7\%                        
\end{tabular}
\end{table}

A large advantage of the diffusion-based second-order estimator lies in its computation speed, which scales extremely well with the number of currencies. In Table~\ref{table:speed} we compare the computation times of the diffusion-based and the mean-reversion-based estimator by adding additional collateral spreads with randomized parameters of the same magnitude as given in Table~\ref{table:spreadparams}.
In the base-case of $3$ currencies, we have computation times of  $1.1$ and $2.5$ seconds, for the diffusion-based and mean-reversion-based estimator, respectively. These computations are performed on an ordinary consumer device where the calculations over the time-discretization $t_k \in \mathcal{T}$ are handled sequentially, this can be fully parallelized. 
Once the moments of the common factor maximum $\widetilde M(t_k)$ are computed for all interpolation times $t_k\in\mathcal{T}$, the diffusion-based estimator $\mathrm{CF}_2^{(1)}$ requires only a constant number of additional operations. The computation time of the common factor maximum moments increases linearly with the number of currencies, since for each currency added, only the cumulative distribution function of the independent maximum, defined in \eqref{eq:cdfmaxAi} by
\begin{equation}
F_{\max_i(A_i(t_k))}(x) = \prod\limits_{i=1}^N \Phi\Bigl(\frac{x-\mu_i(t_k)}{\sqrt{\sigma_i^2(t_k) - \sigma_{\min}^2(t_k)|\gamma(t_k)|}}\Bigr),\ x\in\R,
\end{equation}
has one factor added. The steep decline in proportion of computation time used on the moments from $3$ to $4$ currencies is explained by the time spent on computing the common factor parametrization. In the case of $3$ currencies, the correlation optimization parameter $\gamma(t_k)$ can be computed analytically (at a proportion of 0.8\% of the diffusion-estimators computation time), whereas from $4$ currencies onwards a numerical solver is needed which takes a proportion of around 5\% computation time in this example.
In contrast, computing the mean-reversion-based estimator $\mathrm{CF}_2^{(2)}$ get significantly more demanding as the number of additional currencies increases. This is rooted in the common factor maximum probabilities of Lemma~\ref{lem:maxprobs}: whenever an additional collateral spread is added, not only do we need to compute the probability over time that this new spread is the maximal spread, also computing the probabilities with which the already existing spreads are the maximal ones gains complexity.
%
%
%
\section{Conclusion}
Translating FX-adjusted collateral rates to collateral spreads is an effective method to reduce the model volatilities, which allows for precise pricing of the collateral choice option with a second order approximation. 
By imposing conditional independence on the collateral spreads, by means of a common factor approximation, a semi-analytical solution is obtained for the moments of the maximum of the collateral spreads. From these, precise second-order estimators are developed. In particular the diffusion-based second order common factor estimator admits precise results which can be computed very efficiently, particularly for a large number of available currencies.
%
\section*{Acknowledgements}
We thank two anonymous referees for their valuable comments, which helped improve this article.
This research is part of the ABC--EU--XVA project and has received funding from the European Union's Horizon 2020 research and innovation programme under the Marie Sk\l{}odowska--Curie grant agreement No.\ 813261. We are grateful to the participants of ABC--EU--XVA for their helpful feedback on the content of this article.

\bibliography{cf-Lit}

%
\appendix
\section{Model extension}\label{sec:appx1}
Assume that the correlation spreads $(q_1, \dots, q_N)$ can be ordered into two groups, $(q_1, \dots, q_n)$ being group 1 and $(q_{n+1}, \dots, q_N)$ being group 2.
Then, for all $t_k \in \mathcal{T}$, define the common factor approximations by group affiliation,
\begin{align}
&\widetilde q_i(t_k) = C^{(1)}(t_k) + A_i^{(1)}(t_k), &i\in\{1,\dots,n\}, \\
&\widetilde q_j(t_k) = C^{(2)}(t_k) + A_j^{(2)}(t_k), &j\in\{n+1,\dots,N\},
\end{align}
where all involved random variables are independent, except for the common factor pair $(C^{(1)}(t_k), C^{(2)}(t_k))$, which follows a specified joint distribution.
In the resulting model, common factor approximations from the same group will be positively correlated with the usual covariance, 
\begin{align}
&\Cov(\widetilde q_i(t_k), \widetilde q_\ell(t_k)) = \Var[C^{(1)}(t_k)], &\text{ for } 1\leq i,\ell \leq n, \\
&\Cov(\widetilde q_j(t_k), \widetilde q_m(t_k)) = \Var[C^{(2)}(t_k)], &\text{ for } n<j,m \leq N.
\end{align}
If they stem from different groups, the covariance will now be determined by the joint distribution of $(C^{(1)}(t_k), C^{(2)}(t_k))$, which allows for negative covariances:
\begin{equation}
\Cov(\widetilde q_i(t_k), \widetilde q_j(t_k)) = \Cov[C^{(1)}(t_k), C^{(2)}(t_k)], \text{ for } 1\leq i\leq n < j \leq N.
\end{equation}
The maximum of this extended common factor model still offers a decomposition akin to \eqref{eq:indep-max}, from which a cumulative distribution function is analytically available. Denote the maxima over the idiosyncratic factors by
\begin{align}
H^{(1)}(t_k) &:= \max(A_1^{(1)}(t_k), \dots, A_n^{(1)}(t_k)), \\
H^{(2)}(t_k) &:= \max(A_{n+1}^{(2)}(t_k), \dots, A_N^{(2)}(t_k)).
\end{align}
Then, the common factor maximum can be rewritten as
\begin{align}
\max(0, \widetilde q_1(t_k), \dots, \widetilde q_N(t_k)) &= \max(0, \max(\widetilde q_1(t_k), \dots, \widetilde q_n(t_k)), \max(\widetilde q_{n+1}(t_k), \dots, \widetilde q_N(t_k))) \nonumber \\
&= \max\bigl(0, C^{(1)}(t_k) + H^{(1)}(t_k), C^{(2)}(t_k) + H^{(2)}(t_k)\bigr) \label{eq:extendedcfmax}
\end{align}
Expression \eqref{eq:extendedcfmax} still admits a cumulative distribution function:
\begin{align}
&\P\Bigr[\max\bigl(0,  C^{(1)}(t_k) + H^{(1)}(t_k), C^{(2)}(t_k)+ H^{(2)}(t_k)\bigr) \leq z\Bigr] \nonumber \\
&=	\begin{cases}
	0, & z < 0, \\
	\P\Bigr[\max\bigl(C^{(1)}(t_k) + H^{(1)}(t_k), C^{(2)}(t_k)+ H^{(2)}(t_k)\bigr) \leq z\Bigr], & z \geq 0.
	\end{cases}
\end{align}
Further, it holds that
\begin{align}
&\P\Bigr[\max\bigl(C^{(1)}(t_k) + H^{(1)}(t_k), C^{(2)}(t_k)+ H^{(2)}(t_k)\bigr) \leq z\Bigr] \nonumber\\
& = \P\Bigr[\bigcup_{x_1,x_2\in\R} \Bigl\{C^{(1)}(t_k) \leq x_1,\ H^{(1)}(t_k) \leq z-x_1,\ C^{(2)}(t_k) \leq x_2,\ H^{(2)}(t_k) \leq z-x_2\Bigr\}\Bigr] \nonumber\\
& = \int\limits_\R \int \limits_\R \Biggl( \int_{-\infty}^{z-x_1} \int_{-\infty}^{z-x_2} f_{(C^{(1)}(t_k), C^{(2)}(t_k))}(x_1, x_2) f_{H^{(1)}(t_k)}(u) f_{H^{(2)}(t_k)}(v) \d u \d v \Biggr) \d x_1 \d x_2 \nonumber\\
& = \int\limits_\R \int \limits_\R f_{(C^{(1)}(t_k), C^{(2)}(t_k))}(x_1, x_2) F_{H^{(1)}(t_k)}(z-x_1) F_{H^{(2)}(t_k)}(z-x_2) \d x_1 \d x_2.
\end{align}
Here, $f_{(C^{(1)}(t_k), C^{(2)}(t_k))}$ denotes the pre-specified density of the common factors and $F_{H^{(1)}(t_k)}$, $F_{H^{(2)}(t_k)}$ are the cumulative distribution functions of the maximum over the idiosyncratic factors, which can be obtained from Lemma~\ref{lem:max-indep}.
\par
In general, this extension is not limited to two groups, allowing for further correlation structures at the cost of an increasingly sophisticated cumulative distribution function of the common factor maximum.

\section{Algorithm}\label{sec:appxalgo}
We outline the necessary computations to obtain the first and second order common factor approximations. Note that the diffusion and mean-reverting estimators can be computed independently from another.

\fbox{
\parbox{\textwidth}{
Inputs:
\begin{enumerate}
\setlength{\itemsep}{0pt}
\setlength{\parskip}{0pt}
	\item Maturity $T$ and number of additional currencies $N$
	\item Time discretization $\mathcal{T} = \{0 = t_0 < t_1 < \dots < t_R = T\}$
	\item Hull--White parametrization of collateral spreads: $(\kappa_i, \xi_i, \theta_i(t_k), q_i(0))$ for all $i\leq N$ and $t_k \in \mathcal{T}$
	\item Instantaneous correlations $(\rho_{i,j})$ for all $i,j \leq N$.
\end{enumerate}
Common factor approximation:
\begin{enumerate}
\setlength{\itemsep}{0pt}
\setlength{\parskip}{0pt}
\setcounter{enumi}{4}
	\item Common factor parameters for every $t_k \in \mathcal{T}$:
	\begin{enumerate}
	\setlength{\itemsep}{0pt}
	\setlength{\parskip}{0pt}
		\item Compute correlation parameter $\gamma(t_k)$ (from \eqref{eq:optimalgamma})
		\item Compute means and variances of $C(t_k), A_i(t_k)$ for all $i\leq N$ (from \eqref{eq:C}-\eqref{eq:Ai})
	\end{enumerate}
	\item Moments of the maximum for every $t_k \in \mathcal{T}$:
	\begin{enumerate}
	\setlength{\itemsep}{0pt}
	\setlength{\parskip}{0pt}
		\item Compute convolution $(F_{\max_i(A_i(t_k))} * f_{C(t_k)})(x)$ over domain $(0, L)$ large enough that $(F_{\max_i(A_i(t_k))} * f_{C(t_k)})(L)\approx 1$
		\item Compute $E[\widetilde M(t_k)] = \int_0^\infty (1-(F_{\max_i(A_i(t_k))} * f_{C(t_k)})(x) \d x$
		\item Compute $\E[\widetilde M(t_k)^2] = \int_0^\infty 2x(1-(F_{\max_i(A_i(t_k))} * f_{C(t_k)})(x) \d x$
	\end{enumerate}

	\item Second order estimators
	\begin{enumerate}
	\setlength{\itemsep}{0pt}
	\setlength{\parskip}{0pt}
		\item Diffusion estimator
		\begin{enumerate}
		\setlength{\itemsep}{0pt}
		\setlength{\parskip}{0pt}
			\item: Compute $\Psi(T) = \Var[\int_0^T X(t)\d t]$ from \eqref{eq:intvarX}
		\end{enumerate}

		\item Mean-reverting estimator:
		\begin{enumerate}
		\setlength{\itemsep}{0pt}
		\setlength{\parskip}{0pt}
		\item Compute $\P[\widetilde M(t_k) = q_i]$ for all $i\leq N$ and $t_k\in\mathcal{T}$ with \eqref{eq:maxprobs}
		\item Compute $\chi(T)$ estimator from \eqref{eq:sanko-variance}
		\end{enumerate}
	\end{enumerate}
\end{enumerate}
Outputs:
\begin{enumerate}
\setlength{\itemsep}{0pt}
\setlength{\parskip}{0pt}
\setcounter{enumi}{7}
	\item First order estimator: $\mathrm{CF}_1(T) = \exp(-\sum_{t_k} \E[\widetilde M(t_k)] \Delta t_k)$
	\item Second order estimators
	\begin{enumerate}
	\setlength{\itemsep}{0pt}
	\setlength{\parskip}{0pt}
		\item Diffusion based: 	$\mathrm{CF}_2^{(1)}(T) = \mathrm{CF}_1(T) (1-\frac12 \Psi(T))$
		\item Mean-reversion based: $\mathrm{CF}_2^{(2)}(T) = \mathrm{CF}_1(T) (1-\frac12 \chi(T))$
	\end{enumerate}
\end{enumerate}
}}  

\section{From rate to spread parameters}\label{sec:appxparams}
In the following, we show how collateral rate parameters can be translated to collateral spread parameters. 
To this end, let $r_0$ and $r_1$ be the collateral rates of two currencies with Hull--White dynamics
\begin{align}
\d r_\ell(t) &= \kappa^{r}_\ell \bigl(\theta^{r}_\ell(t) - r_\ell(t)\bigr)\d t + \xi^{r}_\ell \d W^{r}_\ell(t),\\
r_\ell(0) &= r_{\ell,0},\\
\d [W_0(t), W_1(t)] &= \rho^{r}_{0,1}\d t,
\end{align}
where the superscript $r$ underlines that these are the Hull--White parameters of FX-adjusted collateral rates and $\ell\in\{0,1\}$. We aim to explain the choice of our parameters of the difference processes, namely the collateral spreads in Table~\ref{table:spreadparams},
\begin{multline}
\d (r_1 - r_0)(t) = \Bigl(\kappa_1^{r}\bigl(\theta_1^{r}(t) - r_1(t)\bigr) - \kappa_0^{r}\bigl(\theta_0^{r}(t) -r_0(t)\bigr)\Bigr)\d t + \Bigl(\xi_1^{r}\d W_1^{r}(t) - \xi_0^{r} \d W_0^{r}(t)\Bigr).
\end{multline}
The drift part of the dynamics can be expressed as
\begin{equation}
\Bigl(\bigl(\kappa_1^{r}\theta_1^{r}(t) - \kappa_0^{r}\theta_0^{r}(t)\bigr) - \frac{\kappa_0^r + \kappa_1^r}{2} \bigl(\frac{2\kappa_1^r}{\kappa_0^r + \kappa_1^r}r_1(t) - \frac{2\kappa_0^r}{\kappa_0^r + \kappa_1^r}r_0(t)\bigr) \Bigr)\d t
\end{equation}
which, if the mean reversion speeds $\kappa_0^r$, $\kappa_1^r$ are similar, is closely approximated by
\begin{equation}
\Bigl(\bigl(\kappa_1^{r}\theta_1^{r}(t) - \kappa_0^{r}\theta_0^{r}(t)\bigr) - \frac{\kappa_0^r + \kappa_1^r}{2} \bigl(r_1(t) - r_0(t)\bigr) \Bigr)\d t.
\end{equation}
This corresponds to the drift part of a Hull--White process with averaged speed of mean reversion $(\kappa_0^r + \kappa_1^r)/{2}$.
Analogously, for the volatility term it holds approximately
\begin{equation}
\xi_1^r \d W^r_1(t) - \xi_0^r \d W^r_0(t) \approx \frac{\xi_1^r + \xi_0^r}{2} \d \bigl(W_1^r(t) - W_0^r(t)\bigr),
\end{equation}
when the volatility parameters are similar. The difference of correlated Brownian motions can be expressed as an independent, scaled Brownian motion $U$,
\begin{equation}
\d \bigl(W_1^r(t) - W_0^r(t)\bigr) = \sqrt{2-2\rho^r_{0,1}} \d U(t).
\end{equation}
In summary, the difference process $q_1(t) = (r_1 - r_0)(t)$ can be well approximated by a Hull--White process with speed of mean reversion parameter $\kappa_1 = (\kappa_1^r+\kappa_0^r)/2$ and volatility parameter $\xi_1 = \sqrt{2-2\rho^r_{0,1}}(\xi_1^r + \xi_0^r)/2$, given that the collateral rate parameters are not largely different. The initial spread value is directly obtained from the difference of initial values,  $q_1(0) = r_1(0)-r_0(0)$.
Using this approximation, we translate the estimated parameters of \cite{SankovichZhu} (Table~\ref{table:rateparams}) to exemplary spread parameters (Table~\ref{table:spreadparams}). It is important to note that the collateral rates are parametrized under the $\Q_0$ measure, whilst the collateral spreads are parametrized under the $T$-forward measure. This change of measure is expressed by a shift in the drift term, which can be included in the long-term means $\theta_i(t)$. 
\begin{table}[htb]
\centering
\caption{Rate parameters}\label{table:rateparams}
\begin{tabular}{ll|ll|ll|ll}
$\kappa^r_0$ & $0.0072$ & $\xi^r_0$ & $0.0073$ & $r_0(0)$ & $0.000845$ & $\rho^r_{0,1}$ & $0.97$\\
$\kappa^r_1$ & $0.0083$ & $\xi^r_1$ & $0.0073$ & $r_1(0)$ & $0.001514$ & $\rho^r_{0,2}$ & $0.95$\\
$\kappa^r_2$ & $0.0080$ & $\xi^r_2$ & $0.0074$ & $r_2(0)$ & $0.002265$ & $\rho^r_{1,2}$ & $0.95$
\end{tabular}
\end{table}


\end{document}